\documentclass[10pt,a4paper,reqno]{amsart}
\usepackage{amsthm,amssymb,amstext,graphicx,comment}
\usepackage[foot]{amsaddr}
\usepackage[a4paper, margin=2.5cm]{geometry}
\usepackage{bbm}
\usepackage{enumerate}
\usepackage{amscd}
\usepackage{mathtools}
\usepackage{enumitem}
\usepackage{empheq}
\usepackage[utf8]{inputenc}
\usepackage[T1]{fontenc}
\usepackage{lmodern}
\usepackage{xcolor}
\usepackage{bibentry}
\usepackage{soul}
\usepackage{natbib}
\setcitestyle{numbers,comma,open={[},close={]}}
\usepackage[colorlinks=true, linkcolor=blue, citecolor=blue, urlcolor=blue]{hyperref}

\DeclareMathOperator\dom{dom}

\newcommand*\D{\mathop{}\!\mathrm{d}}

\newcommand*\E{\mathop{}\!\mathrm{e}}

\DeclareMathOperator{\tr}{Tr}

\numberwithin{equation}{section}
\newtheorem{theorem}{Theorem}[section]
\newtheorem*{theorem*}{Theorem}
\newtheorem{lemma}[theorem]{Lemma}
\newtheorem{proposition}[theorem]{Proposition}
\newtheorem*{proposition*}{Proposition}

\theoremstyle{plain}

\theoremstyle{definition}
\newtheorem{definition}[theorem]{Definition}
\newtheorem{remark}[theorem]{Remark}

\newtheoremstyle{example}
  {.3\baselineskip}
  {.3\baselineskip}
  {\normalsize}  
  {0pt}       
  {\bfseries} 
  {.}         
  {5pt plus 1pt minus 1pt} 
  {}          
\theoremstyle{example}

\newtheorem*{assumption*}{\assumptionnumber}
\providecommand{\assumptionnumber}{}
\makeatletter
\newenvironment{assumption}[2]
 {%
  \renewcommand{\assumptionnumber}{Assumption #1$\mathfrak{#2}$}%
  \begin{assumption*}%
  \protected@edef\@currentlabel{#1$
  \mathfrak{#2}$}%
 }
 {%
  \end{assumption*}
 }
 
\makeatother

\def\eps{\varepsilon}
\renewcommand{\epsilon}{\eps}

\renewcommand{\MR}{\mathbb{R}}
\newcommand{\MC}{\mathbb{C}}

\newcommand{\MN}{\mathbb{N}}
\newcommand{\MP}{\mathbb{P}}
\newcommand{\MQ}{\mathbb{Q}}

\newcommand{\R}{\MR}
\newcommand{\N}{\MN}

\newcommand{\cF}{\mathcal{F}}
\newcommand{\cA}{\mathcal{A}}
\newcommand{\cB}{\mathcal{B}}

\newcommand{\cH}{\mathcal{H}}

\newcommand{\cL}{\mathcal{L}}

\newcommand{\cO}{\mathcal{O}}

\newcommand{\cR}{\mathcal{R}}
\newcommand{\cV}{\mathcal{V}}

\newcommand{\cY}{\mathcal{Y}}

\newcommand{\cS}{\mathcal{S}}

\newcommand{\sR}{\mathsf{R}}

\newcommand{\sY}{\mathsf{Y}}

\newcommand{\fA}{\mathfrak{A}}

\newcommand{\df}{\coloneqq}

\newcommand{\interior}[1]{({\kern0pt#1})^{\textnormal{o}}}
\newcommand{\set}[1]{\left\{ #1\right\}}
\newcommand{\norm}[1]{\|#1\|}

\newcommand{\EX}[1]{\mathbb{E}\left[#1\right]}

\newcounter{Task}
\newcommand{\sgc}{\begin{color}{blue}}
\newcommand{\cgs}{\end{color}}
\newcommand{\cHplus}{\cH^{+}}

\newcommand{\cHpluso}{\cHplus\setminus \{0\}}

\newcommand{\dm}{m(\D\xi)}
\newcommand{\dmu}{\mu(\D\xi)}

\begin{document}
\title[Pricing Options on Forwards in Affine SV Models]{Pricing Options on Forwards in
Function-Valued Affine Stochastic Volatility Models} 
\address{\today{}, Korteweg-de Vries Institute for Mathematics at University of Amsterdam.}

\thanks{Asma Khedher is grateful for the financial support by the Research Foundation Flanders (FWO) under the grant FWO WOG W001021N}

\maketitle{}

\vspace{-3mm}
\begin{center}
  \begin{tabular}{ccc}
  \textsc{Jian He}  &  \textsc{Sven Karbach} &  \textsc{Asma Khedher}\\
  \small[\texttt{\MakeLowercase{j.he2@uva.nl}}] &    
   \small[\texttt{\MakeLowercase{sven@karbach.org}}] &
   \small[\texttt{\MakeLowercase{a.khedher@uva.nl}}]
  \end{tabular}
\end{center}

\begin{abstract}
We study the pricing of European-style options written on forward contracts within function-valued infinite-dimensional affine stochastic volatility models. The dynamics of the underlying forward price curves are modeled within the Heath-Jarrow-Morton-Musiela framework as solution to a stochastic partial differential equation modulated by a stochastic volatility process. We analyze two classes of affine stochastic volatility models: (i) a Gaussian model governed by a finite-rank Wishart process, and (ii) a pure-jump affine model extending the Barndorff–Nielsen–Shephard framework with state-dependent jumps in the covariance component. For the pure-jump model, we derive conditions for the existence of exponential moments and develop semi-closed Fourier-based pricing formulas for vanilla call and put options written on forward contracts with fixed delivery lag. For the Wishart model, we develop finite-rank Riccati approximations and corresponding Fourier-based pricing approximations. Our approach allows for tractable pricing in function-valued forward-curve models while capturing maturity-specific and term structure risk essential in forward markets.\\


\end{abstract}


\section{Introduction}

The aim of this work is to derive semi-closed pricing formulas for European vanilla options written
on forward contracts in a function-valued HJMM model with stochastic volatility, and to validate
these formulas numerically. Our pricing results are formulated for option families with a fixed
delivery lag. To achieve this, we adopt the Heath--Jarrow--Morton--Musiela (HJMM) paradigm,
originally proposed for interest-rate markets by \cite{heath1992bond} and later transferred to
commodity markets by \cite{clewlow2000energy,benth2008stochastic}. In this paradigm, forward prices are modeled directly in time-to-maturity coordinates. Under this representation, the forward curve evolves according to a hyperbolic stochastic partial differential equation (SPDE) driven by a cylindrical Brownian motion whose covariance is itself modeled by a stochastic process that we call the \emph{instantaneous covariance process}. We investigate two \emph{affine} specifications of this covariance
process: (i) infinite-dimensional Wishart processes
taking values in the cone of positive self-adjoint trace class operators on
a separable real Hilbert space, as introduced in \cite{cox2024infinite}, and (ii) a state-dependent jump extension of the operator-valued Ornstein--Uhlenbeck process introduced in~\cite{CKK22b}, in which the jump intensity measure depends affinely on the current covariance state. The pure-jump model admits semi-explicit Fourier-Laplace expressions in terms of \emph{generalized Riccati equations}, while for the Wishart model we work with finite-rank Riccati approximations. Exploiting this structure, together with establishing rigorous conditions for the existence of exponential moments in the pure-jump setting, enables the derivation of semi-closed pricing formulas for vanilla call and put options and finite-rank pricing approximations in the Wishart setting.

\subsection{The HJMM framework}
For any time $t \geq 0$ and maturity date $T \geq t$, the forward price $F(t,T)$ represents the agreed-upon price at time $t$ for a transaction occurring at the future date $T$. Utilizing the Musiela parametrization $x \equiv T - t$, forward prices become functions of time-to-maturity. We define the logarithmic forward prices as $f_{t}(x)\equiv \ln(F(t,t+x))$, where $x,t\geq 0$, and call the mapping $x \mapsto f_t(x)$ the (logarithmic) \emph{forward price curve} at time $t \geq 0$. The forward curve dynamics are described by the stochastic partial differential equation (SPDE):
\begin{align}\label{eq:HJMM}
\begin{cases}
\D f_t(x) &= \left(\frac{\partial}{\partial x}f_t(x) + g_t(x)\right)\D t + \sum_{i=1}^{d}\sigma_t^{(i)}(x)\D W_t^{(i)},\\
f_0(x) &= \ln(F(0,x)),
\end{cases}
\end{align}
where $(W^{(i)}_t)_{0 \leq t \leq T}$, ${i=1}, \ldots, {d}$, are independent Brownian motions. The
drift term $g_t(x)$ satisfies a \emph{HJM no-arbitrage condition} (see~\cite[Sec.~2.4]{CT06} and
Section~\ref{sec:appl-math-finance} below), ensuring that each forward price process $(F(t,T))_{0 \leq t \leq T}$ is a martingale under some risk-neutral measure $\MQ$. 

The functions $(\sigma_t^{(i)}(x))_{0 \leq t \leq T}$, ${i=1}, \ldots, {d}$, are referred to as the \textit{instantaneous volatilities} of the forward curve. They characterize the local covariance structure between forward prices of different maturities. The integer $d \in \mathbb{N} \cup \{+\infty\}$ is called the \emph{rank of the noise} and corresponds to the number of independent risk factors driving the forward curves. If $d < \infty$, we refer to~\eqref{eq:HJMM} as a \emph{finite-rank HJM} model; otherwise, it is an \emph{infinite-rank HJM} model. 

\subsection{Stochastic covariance models for forward curve dynamics}

The diffusion term in~\eqref{eq:HJMM} can be interpreted in a more abstract form as the stochastic integral of an operator-valued process $(\sigma_t)_{0\leq t \leq T}$ acting on the increments of a cylindrical Brownian motion $(\D W_t)_{0\leq t \leq T}$. We understand $(\sigma_t)_{0\leq t \leq T}$ as the stochastic volatility operator, encoding the time-varying covariance structure of the forward curve dynamics. The functions $\sigma_t^{(i)}$ can therefore be interpreted as the principal components of $\sigma_t$, capturing the volatility structure along different directions in the underlying space of forward curves. The joint evolution of the forward curve and its stochastic volatility, represented by the process $(f_t, \sigma_t)_{0\leq t \leq T}$, thus defines an \emph{infinite-dimensional stochastic volatility model} for forward curve dynamics.

The motivation for adopting a stochastic volatility framework stems from empirical evidence in commodity markets: volatility is time-varying, exhibits clustering, and is prone to abrupt spikes; see, e.g., \cite{eydeland2002energy, benth2012modeling}. Classical models with constant volatility cannot capture these phenomena, whereas stochastic volatility models allow for a more realistic and flexible description of market behavior. Furthermore, to accurately capture volatility smiles in commodity option markets, incorporating stochastic volatility is essential. In contrast to approaches that perform a priori dimension reduction of the covariance, the pure-jump specification is formulated directly on the full maturity space rather than on a PCA-truncated factor model. For the Wishart specification considered here, the covariance state remains finite rank for fixed initial rank, and our analytical and numerical treatment is therefore based on finite-rank Riccati approximations. This still preserves the infinite-dimensional maturity state space while keeping the Gaussian covariance dynamics implementable. Empirical studies support the relevance of high-dimensional maturity dynamics: for instance, \cite{andresen2010modeling} and \cite{koekebakker2005forward} demonstrate that power forward markets exhibit substantial idiosyncratic risk, and that high-dimensional noise is needed to capture the observed variability. In particular, the latter shows that even in PCA, around 10 factors are typically required to explain the majority of the variance, underlining the importance of accommodating rich maturity dynamics from the outset. Here our function-valued formulation helps, as it is essentially non-parametric, relying on the full initial forward curve and the full covariance operator as model inputs. The covariance structure is encoded via a positive kernel operator, which can capture empirical features with a minimal set of hyperparameters and enables direct calibration to market data, including option prices and historical term structure dynamics.

The stochastic volatility models considered in this paper belong to the class of \emph{affine processes}, which play a central role in modern financial modeling due to their analytical tractability. Affine processes satisfy the so-called \emph{affine transform formula}, which expresses the cumulant generating function as an affine function of the initial state and the solution of a system of generalized Riccati ordinary differential equations (ODEs). A key advantage of this framework is the quasi-explicit availability of the Fourier–Laplace transform, which enables efficient and accurate pricing of options and other derivatives via Fourier inversion methods; see~\cite{CM99, DFS03, eberlein2010analysis}. The extension of affine modeling techniques to infinite-dimensional settings has been the subject of extensive research; see, e.g.,~\cite{schmidt2020infinite, cuchiero2020generalized, benth2018ornstein, benth2018heston, Kar24, Kar25, FK24, cuchiero2025measure, benth2025measure}. In parallel, robust statistical methods have been developed for the estimation and analysis of stochastic volatility and covariation in term structure models. These approaches address both the functional nature of the data and the dynamic consistency required for sound modeling; see~\cite{schroers2024robustfunctionaldataanalysis, schroers2024dynamicallyconsistentanalysisrealized}. 

\subsection{Contribution} In this paper, we contribute to this growing literature by developing a pricing framework for European options written on forward contracts under function-valued affine stochastic volatility. We formulate the forward curve dynamics within the HJMM framework and model the instantaneous covariance as an operator-valued affine process. The resulting risk-neutral pricing problem is treated for a fixed delivery lag. Two concrete model classes are considered: a Gaussian specification driven by a Wishart covariance process on the cone of positive trace-class operators, and a pure-jump model based on an extension of an operator-valued Ornstein–Uhlenbeck process to accommodate state-dependent jumps. The inclusion of jumps in the volatility process enhances the model’s ability to reflect sudden market shocks, which commonly arise from supply disruptions, geopolitical tensions, or abrupt changes in demand; see, e.g.~\cite{leippold2010asset}. For the pure-jump specification, we rigorously extend the affine transform formula to a complex strip in an infinite-dimensional Hilbert space and establish sufficient conditions under which exponential moments exist in real and complex domains. In particular, we identify a sufficient admissible domain for the moment generating function using generalized Riccati equations in Banach spaces, building on the analytic machinery introduced in~\cite{Mar76}, and extending prior finite-dimensional results such as those in~\cite{KM15}. For the Wishart specification, we work with finite-rank Riccati systems and the resulting Fourier-based pricing approximations.

Based on these results, we derive semi-closed-form pricing formulas for vanilla call and put options in the pure-jump model, together with finite-rank pricing approximations in the Wishart model, which enable numerical evaluation via Fourier inversion. We validate the accuracy and robustness of this pricing methodology through comparisons with direct or conditional Monte Carlo benchmarks. The numerical results show that the first forward moment in the Wishart experiment, the corrected Wishart option prices, and the option prices in the pure-jump experiments are consistent with these benchmarks, see Figures~\ref{fig:mmt_ws}, \ref{fig:wishart_pricing}, \ref{fig:mc_b0}, and \ref{fig:mc_b1} below. For the explicit L\'evy and BNS specifications, the affine formulas provide substantial computational gains over Monte Carlo simulation. In the state-dependent jump example, the relative runtime ranking is much more implementation-dependent and should therefore not be overinterpreted. Our results demonstrate that, despite the infinite-dimensional nature of the maturity state space, affine stochastic volatility models remain useful tools for capturing rich dynamics in forward and futures markets, provided the tractability claims are interpreted in line with the particular specification under consideration.

\subsection{Layout of the article}

The remainder of the paper is organized as follows. In Section~\ref{sec:hjmm}, we introduce the modeling framework for forward curve dynamics using the HJMM approach and define the class of infinite-dimensional stochastic volatility models under consideration. Sections~\ref{sec:OU-jumps} and~\ref{sec:OU-Wishart} present two concrete affine specifications for the volatility process: a pure-jump model based on an extension of an operator-valued Ornstein–Uhlenbeck process, and a Gaussian model driven by a Wishart process on the cone of positive trace-class operators. In Section~\ref{sec:appl-math-finance}, we provide Fourier-based pricing formulas for the pure-jump model and finite-rank pricing approximations for the Wishart model. Section~\ref{sec:numerics} presents a numerical study comparing these affine formulas with Monte Carlo simulation methods. We highlight the accuracy, robustness, and computational efficiency of the affine approach, especially in settings where simulation becomes numerically challenging. Finally, in Section~\ref{sec:complex-moments}, we study the existence of the exponential moments.

\subsection{Notation}
For $(X,\tau)$ a topological space and $S \subset X$ we let $\cB(S)$ denote the Borel-$\sigma$-algebra generated by the relative topology on $S$. We denote by $C^k([0,T];S)$ the space of $S$-valued $k$-times continuously differentiable functions on $[0,T]$. Throughout this article we fix a separable, infinite-dimensional real Hilbert space $(H,\langle\cdot,\cdot \rangle_H)$. The complexification of a real Hilbert space is the complex vector space $H^{\mathbb{C}}= H\oplus_\R \mathrm{i} H=\{x+\mathrm{i}y \colon x,y \in H\}$ equipped with the complex \emph{bilinear} extension of the real inner product,
$\langle x+\mathrm{i}y, u+\mathrm{i} v\rangle_{H^{\mathbb{C}}} = \langle x,u \rangle_{H}-\langle y,v\rangle_{H}+\mathrm{i}\langle y,u \rangle_{H}+\mathrm{i}\langle x, v\rangle_{H}$, for all $x,y,u,v \in H$. Note that $\langle\cdot,\cdot\rangle_{H^{\mathbb{C}}}$ is bilinear (not sesquilinear) and hence \emph{not} a Hilbert-space inner product; it is the natural pairing for the analytic continuation of the Fourier--Laplace transform used throughout this paper. For $z =x+\mathrm{i} y \in H^{\mathbb{C}}$, we say that $x=\Re(z)$ and $y=\mathcal{I}(z)$ are the real and imaginary parts of $z$. Moreover, $\bar{z}=\Re(z)- \mathrm{i} \,\mathcal{I}(z)$ 
denotes the conjugate of $z$. \par
The space of bounded linear operators from $H$ to $H$ is denoted by $\cL(H)$. The space of compact operators from $H$ to $H$ is denoted by $(K(H),\left\| \cdot \right\|_{\cL(H)})$. The adjoint of an
operator $A \in \cL(H)$ is denoted by $A^*$, i.e.,  $\langle Ah,g\rangle_{H} = \langle h, A^*g\rangle_{H}$ for all $h\in H$.
For $A\in \cL(H^{\mathbb{C}})$
we define the transpose $A^T\in \cL(H^{\mathbb{C}})$
by $\langle A^T h, g\rangle_{H^{\mathbb{C}}} = \langle A \bar{g}, \bar{h} \rangle_{H^{\mathbb{C}}}$, $h,g\in H^{\mathbb{C}}$. Note that in contrast to $A^*$, the operator $A^{T}$ corresponds to the transpose without conjugation.
Note that any $A\in \cL(H)$ extends in a trivial and norm-conserving way to an operator $\tilde{A}\in \cL(H^{\mathbb{C}})$ by setting $\tilde{A}(h + ig)=Ah+iAg$, $h,g\in H$. For all $p\in [1,\infty)$ let 
$(\cL_{p}(H), \left\| \cdot \right\|_{\cL_p(H)})$ be the Banach space of Schatten class operators from $H$ to $H$, i.e.,
\begin{equation}
 \cL_{p}(H) = \Big\{ A \in K(H) \colon \sum_{\lambda \in \sigma(A^*A)} \lambda^{p/2} < \infty \Big\}\,,
\end{equation}
and $\| A \|_{\cL_p(H)}^{p} = \sum_{\lambda \in \sigma(A^*A)} \lambda^{p/2}$, where eigenvalues are counted with multiplicity.
In particular, we let $\cL_{1}(H)$ and $\cL_{2}(H)$ denote respectively the
space of \emph{trace class operators} and the space of \emph{Hilbert-Schmidt operators} on $H$. Recall 
that $\cL_{2}(H)$ is a Hilbert space when endowed with the inner product
\begin{align*}
 \langle A, B \rangle_{\cL_{2}(H)} = \sum_{n=1}^{\infty} \langle A e_n, B e_n \rangle_H.
\end{align*}

For notational brevity, we reserve $\langle \cdot, \cdot \rangle$ to denote the
inner product on $\cL_{2}(H)$, and $\| \cdot \|_2$ for the norm induced by
$\langle \cdot, \cdot \rangle$. Moreover, we reserve $\| \cdot \|_1$ for the norm $\| \cdot \|_{\cL_{1}(H)}$.
We recall that the trace of $A\in \cL_1(H)$ is defined by
\begin{equation}
 \tr(A) = \sum_{n\in \N} \langle Ah_n, h_n \rangle_{H} \in  \mathbb{C} ,
\end{equation}
where $(h_n)_{n\in \N}$ is an orthonormal basis for $H$; $\tr(A)$ does not depend on the choice of the orthonormal basis. 

Writing $V'$ for the dual of a Banach space $V$, 
we recall (see, e.g.,~\cite[Section 19]{Conway:2000}) that the dual space of trace class operators satisfies $(\cL_1(H))'\simeq \cL(H)$  under the duality pairing
\begin{equation}
 \langle A, B \rangle_{\cL(H),\cL_1(H)}
 = \tr(B^* A),\quad A\in \cL(H),\, B\in \cL_1(H).
\end{equation}

We let $\cS(H)$, $\cS_1(H)$ denote the (closed) subspaces of $\cL(H)$, $\cL_1(H)$ consisting of all operators that are self-adjoint and we let $\cS^+(H)$ and $\cS_1^+(H)$ denote the (closed) subsets of $\cS(H)$ and $\cS_1(H)$ consisting of all self-adjoint operators $A$ satisfying $\sigma(A) \subset [0,\infty)$.
Moreover, if $A\in \cS(H)$, then the extended operator $\tilde{A} \in \cL(H^{\mathbb{C}})$ fulfills $\tilde{A}^T=\tilde{A}=\tilde{A}^*$, in particular, $\tilde{A}\in \cS(H^{\mathbb{C}})$ (from now on we do not distinguish between $A$ and $\tilde{A}$). In this article we will encounter the set $\cS^+(H)\oplus \mathrm{i}\cS(H)\subseteq \cL(H^{\mathbb{C}})$, which
denotes the operators $A\in \cL(H^{\mathbb{C}})$ for which there exist (necessarily unique) $A_1 \in \cS^+(H)$, $A_2\in \cS(H)$ such that $A=A_1 + \mathrm{i} A_2$. The set $\cS(H)\oplus \mathrm{i} \cS(H)$ is defined analogously. Note that if $A\in \cS(H)\oplus \mathrm{i}\cS(H)$, then $A^{T}=A$.

A nonempty subset $K$ of a vector space is called a \emph{wedge} if $K+K\subseteq K$ and $\alpha K \subseteq K$
for all $\alpha \geq 0$, if moreover $K\cap (-K) = \{ 0\}$ then we call $K$ a \emph{cone}. A cone $K$ in a vector space $X$ induces a partial ordering: we 
write $x \leq_{K} y$ if $y-x\in K$ (and $x\geq_K y$ if $x-y\in K$). For a Banach space $V$, if 
$K\subset V$ is a wedge, 
we define the \emph{dual} of $K$ by
\begin{equation}
K^* = \{ x' \in V' \colon \langle x', x\rangle_{V', V} \geq 0 \text{ for all } x\in K \},
\end{equation}
and we say that $K$ is \emph{self-dual} if $K=K^*$. 
We define $\cH$ to be the space of all self-adjoint Hilbert-Schmidt operators on $H$ and 
$\cH^+$ to be the cone of all positive operators in $\cH$:
\begin{equation*}
 \cH := \{ A \in \cL_{2}(H) \colon A = A^* \}, \ \text{and}\ 
 \cH^{+} := \{ A \in \cH \colon \langle Ah, h\rangle_H \geq 0 \text{ for all } h\in H \}. 
\end{equation*}
Note that $\cH$ is a closed subspace of $\cL_{2}(H)$, and that $\cH^{+}$ is a self-dual 
cone in $\cH$.
For
$a,b\in H$, we let $a\otimes b$ be the linear operator defined by
$a\otimes b (h)=\langle a, h\rangle_{H} b$ for every $h\in H$. Note that $a\otimes a\in\cHplus$
for every $a\in H$ and when space is scarce, we shall write $a^{\otimes 2}\df a\otimes a$. Finally, throughout this article we let $\chi\colon \cH\rightarrow \cH$
denote the truncation function given by $\chi(x) = x\mathbbm{1}_{\{\| x \|_2 \leq 1\}}$.

\section{Infinite-Dimensional Affine Stochastic Volatility Models} \label{sec:hjmm}

In this section, we present two affine stochastic volatility models in an infinite-dimensional function-valued setting. In both cases, the underlying forward-curve process solves the HJMM SPDE~\eqref{eq:HJMM} in some separable Hilbert space $H$ of forward curves. The key distinction between the two models lies in the specification of their instantaneous covariance processes $(\sigma_t)_{t\geq 0}$.  

\subsection{The affine pure-jump stochastic volatility model}\label{sec:OU-jumps}
We begin by formally introducing the two instantaneous covariance processes and their respective joint Fourier-Laplace transform.

\subsubsection{A pure-jump $\cH^+$-valued instantaneous covariance process}
To introduce our pure-jump covariance process, we need the following admissibility conditions that we impose on the parameter of the process:
\begin{assumption}{}{A}\label{def:admissibility}
 An \emph{admissible parameter set}
$(b, B, m, \mu)$ 
consists of
\begin{enumerate}
\item[(i)] \label{eq:m-2moment} a measure $m\colon\cB(\cHpluso)\to [0,\infty]$ such that
    \begin{enumerate}
    \item[(a)] $\int_{\cHpluso} \| \xi \|_2^2 \,\dm < \infty$ and
    \item[(b)] $\int_{\cHpluso}|\langle\chi(\xi),h\rangle|\,\dm<\infty$ for all $h\in\cH$  
   and there exists an element $I_{m}\in \cH$ such that $\langle
   I_{m},h\rangle=\int_{\cHpluso}\langle \chi(\xi),h\rangle\, m(\D\xi)$ for every $h\in\cH$\,;
 \end{enumerate}   
\item[(ii)]\label{eq:drift} a vector $b\in\cH$ such that
  \begin{align}\label{ass:b_positive}
    \langle b, v\rangle - \int_{\cHpluso} \langle \chi(\xi), v\rangle\,m(\D\xi) \geq 0\, \quad\text{for all}\;v\in\cHplus\,;
  \end{align}
\item[(iii)] \label{eq:affine-kernel} a $\cH^{+}$-valued measure 
$\mu \colon \mathcal{B}(\cHpluso) \rightarrow \cH^+$ such that
\begin{align*}
\int_{\cH^+\setminus \{0\}} \langle \chi(\xi), u\rangle \frac{\langle \dmu, x \rangle}{\| \xi \|_2^2 }< \infty,  
\end{align*}
for all $u,x\in \cH^{+}$ satisfying $\langle u,x \rangle = 0$\,;
 \item[(iv)] \label{eq:linear-operator} an operator $B\in \mathcal{L}(\mathcal{H})$ 
 with adjoint $B^{*}$ satisfying
\begin{align*}
    \left\langle B^{*}(u) , x \right\rangle 
    - 
    \int_{\cHpluso}
        \langle \chi(\xi),u\rangle 
        \frac{\langle \dmu, x \rangle}{\| \xi\|_2^2 }
    \geq 0,
  \end{align*}
  for all $x,u \in \cHplus$ satisfying $\langle u,x\rangle=0$.
\end{enumerate}
\end{assumption}

In this paper, we consider affine pure-jump covariance processes that are either driven by a pure jump L\'evy process (nullifying the state-dependent jump parts), or assume that the jump-process is of finite-variation, which is due to the existence of exponential moments in these two settings (see Section~\ref{sec:complex-moments} below).

\begin{assumption}{}{B}\label{ass:chkk-finite-variation}
Assume that $(b,B,m, \mu)$ is an admissible parameter such that $\mu(\D\xi)$ takes values in $\cL_{1}(H)\cap\cHplus$. Moreover, assume that either one of the following two cases holds:
\begin{enumerate}
    \item[(i)] $\mu=0$,
\item[(ii)] 
$\int_{\cHplus \cap\{ \norm{\xi}_2\leq 1\}} \|\xi\|_2\,m(\D\xi) < \infty\,$ and 
$\int_{\cHplus \cap \{\norm{\xi}_2\leq 1\}} \|\xi\|_2^{-1}\langle y,\mu (\D
\xi)\rangle < \infty$,\quad $\forall y\in\cHplus.$
\end{enumerate}
\end{assumption}

 Given an admissible parameter set $(b, B, m, \mu)$ that satisfies Assumptions \ref{def:admissibility} and \ref{ass:chkk-finite-variation}, we conclude from \cite[Theorem 2.8]{CKK22a} that there exists a probabilistically and analytically weak solution to
\begin{align}\label{eq:process-sY}
\D \sY_{t} &=\Big(b+B(\sY_{t})+\int_{\cHplus \cap \{\norm{\xi}_2> 1\}}\xi \,M(\sY_t,\D\xi)\Big)\D t+ \D J_{t},   \qquad t\geq 0\,,\\
\sY_0 &= \mathsf{y} \in \cH^+\,,\nonumber
\end{align}
where $(J_t)_{t\geq 0}$ is a purely-discontinuous square integrable $\cH$-valued martingale with jump intensity measure depending on the state $\sY_t$ in affine manner as
\begin{equation}\label{eq:measure-M}
M(y, \D \xi):= m(\D\xi)+ \frac{\langle \mu(\D\xi), y\rangle}{\|\xi\|_2^2}.
\end{equation}


Notice that $(\sY_t)_{t\geq 0}$ is a square integrable $\cH^+$-valued Markovian semimartingale and \eqref{eq:process-sY} amounts to the semimartingale representation of $(\sY_t)_{t\geq 0}$, see \cite[Proposition 2.6]{CKK22b} for the details.

\subsubsection{The affine pure-jump joint stochastic volatility model} 
The first stochastic volatility model we consider takes as instantaneous volatility process the square root of the pure-jump $\cH^+$-valued process from~\eqref{eq:process-sY}. We denote the forward price dynamics by $(X_t)_{t\geq 0}$, governed by the HJMM--SPDE~\eqref{eq:HJMM}, but formulated in a Hilbert space $H$ as follows. Fix an initial value $(x,\mathsf{y})\in H \times \cH^+$ and consider the coupled SDE
\begin{align}\label{eq:joint-affine-SDE}
  \begin{cases}
    \D X_{t} = \big( \cA X_{t} + D^{1/2}\sY_t D^{1/2}\Upsilon \big) \,\D t + D^{1/2}\sY_{t}^{1/2}\,\D W_{t}, \\
    \D \sY_{t} = \Big( b + B(\sY_{t}) + \int_{\cHplus \cap \{\|\xi\|_2 > 1\}} \xi \,M(\sY_t,\D\xi) \Big) \D t + \D J_{t}, \quad t \geq 0 ,
  \end{cases}
\end{align}
where $(W_t)_{t\geq 0}$ is an $H$-cylindrical Brownian motion, $(J_t)_{t\geq 0}$ is as described after equation~\eqref{eq:process-sY}, $(W_t)_{t\geq 0}$ and $(J_t)_{t\geq 0}$ are independent, $(b,B,m,\mu)$ is an admissible parameter set satisfying Assumptions~\ref{def:admissibility} and~\ref{ass:chkk-finite-variation}, $(\cA,\dom(\cA))$ is the generator of a strongly continuous semigroup on $H$ (in applications, the differentiation operator on the Filipovi\'{c} space; see Section~\ref{sec:appl-math-finance}), $D \in \cS_1^+(H)$, and $\Upsilon \in H$. We include the linear term $D^{1/2}\mathsf{Y}_t D^{1/2}\Upsilon$ in the drift of $(X_t)_{t\geq0}$, as motivated by the no-arbitrage drift condition in forward markets;
see Section~\ref{sec:appl-math-finance} for details. This addition does not affect well-posedness or the affine property; see, for example, the proof of \cite[Lemma~2.8]{CKK22b}.

The existence of a (stochastically) weak solution to~\eqref{eq:joint-affine-SDE} in $H \times \cH$ follows from \cite[Proposition~2.6 and Lemma~2.8]{CKK22b}. 
We refer to~\eqref{eq:joint-affine-SDE} as the \emph{pure-jump affine stochastic volatility model}.

\begin{remark}[Independence assumption]\label{rem:independence}
The independence of the driving noise $(W_t)_{t\geq 0}$ of the forward price dynamics and the noise $(J_t)_{t\geq 0}$ (respectively $(B_t)_{t\geq 0}$ in the Wishart model of Section~\ref{sec:OU-Wishart}) of the volatility process is a modeling limitation: it precludes the \emph{leverage effect}, i.e., the empirically observed negative correlation between returns and volatility changes. This assumption is standard in the infinite-dimensional affine framework as it preserves the affine structure. Extending the model to incorporate leverage via correlated driving noises has been done for the BNS-model in \cite{zbMATH07927020}.
\end{remark}

\subsubsection{Riccati equations associated with the pure-jump affine volatility model} \label{sec:Riccati-jumps}

Next we aim to derive the Laplace transform of our joint stochastic volatility model. The
Laplace transform will be computed in terms of a solution to ordinary differential equations
known as \emph{Riccati equations}. We start by introducing the Riccati equations associated with our introduced model.

Let $(b,B,m,\mu)$ be an admissible parameter set satisfying Assumptions \ref{def:admissibility}. Recall the measure $M$ in \eqref{eq:measure-M} and define 
\begin{align*}
  \cO=\set{u\in\cH\colon \int_{\cHplus\cap\set{\norm{\xi}_2>1}}\E^{-\langle
  \xi,u\rangle}\,M(y,\D\xi)<\infty,\quad  \forall y\in\cHplus}.
\end{align*}
We note that $\cO$ is a convex subset in $\cH$ and since
$\cHplus\subseteq \cO$, it is non-empty. Moreover, we assume that $M(x,\D\xi)$
is such that $\cO$ is an open subset of $\cH$.
Define ${\sf F}\colon \cO \to \MR$ and ${\sf R}\colon H \times \cO \to \cH$, respectively as
\begin{align}
 {\sf F}(u)&=\langle \tilde{b}, u\rangle - \int_{\cHpluso}\big(\E^{-\langle \xi, u\rangle}-1\big) m(\D \xi), \label{eq:F}\\
 {\sf R}(h,u)&= \tilde{B}^{*}(u)-\tfrac{1}{2}(D^{1/2}h)^{\otimes 2}
     -D^{1/2} h \otimes D^{1/2}\Upsilon 
    - \int_{\cHpluso}\big(\E^{-\langle
    \xi,u\rangle}-1\big)\frac{\mu(\D \xi)}{\norm{\xi}_2^{2}}, \label{eq:R}
\end{align}
where $\tilde{b}\df b-\int_{\cHpluso}\chi(\xi)\,m(\D\xi)$ and for all $v\in\cHplus$
\begin{align*}
  \tilde{B}(v)\df B(v)-\int_{\cHpluso}\chi(\xi)\frac{\langle
  v,\mu(\D\xi)\rangle}{\norm{\xi}_2^{2}}.    
\end{align*}
Note that $\tilde{b}$ and $\tilde{B}$ are well defined due to the
finite-variation Assumption~\ref{ass:chkk-finite-variation}.
\begin{definition}
Let $T\geq 0$ and $u=(u_{1},u_{2})\in H\times\cO$. Consider the
following \emph{extended generalized Riccati equations}:   
\begin{subequations}\label{eq:extended-Riccati}
  \begin{empheq}[left=\empheqlbrace]{align} 
    \,\frac{\partial {\sf P}}{\partial t}(t,u)&={\sf F}({\sf q}_{2}(t,u)), &\, 0< t \leq T, \quad {\sf P}(0,u)=0,\label{eq:extended-Riccati-P}\\
    \,{\sf q}_{1}(t,u)&=u_{1}+\cA^{*}\left(\int_{0}^{t}{\sf q}_{1}(s,u)\D s\right), &\, 0<
    t \leq T,\quad {\sf q}_{1}(0,u)=u_{1}, \label{eq:extended-Riccati-q1}\\
    \,\frac{\partial {\sf q}_{2}}{\partial t}(t,u)&={\sf R}({\sf q}_{1}(t,u), {\sf q}_{2}(t,u)), &\,
    0< t \leq T, \quad 
    {\sf q}_{2}(0,u)=u_{2}. \label{eq:extended-Riccati-q2}
  \end{empheq}                               
\end{subequations}
For $u=(u_{1},u_{2})\in H\times \cO$ and $T\geq 0$, we say that $({\sf P}(\cdot, u), {\sf q}_1(\cdot,u), {\sf q}_2(\cdot, u)) \colon  [0,T] \rightarrow \mathbb{R}\times H\times \cH$ is a mild solution to \eqref{eq:extended-Riccati-P}-\eqref{eq:extended-Riccati-q2} if 
    ${\sf P}(\cdot,u)\in C^{1}([0,T],\MR)$, ${\sf q}_{1}(\cdot,u)\in C([0,T],H)$ and ${\sf q}_{2}(\cdot,u)\in
    C^{1}([0,T], \cO)$ satisfies \eqref{eq:extended-Riccati-P}-\eqref{eq:extended-Riccati-q2}.
\end{definition}
\begin{remark}[Local well-posedness on $\cO$]\label{rem:forward-invariance}
Local Lipschitz continuity of the map $\mathsf{R}$ on $\cO$ guarantees short-time existence of the Riccati flow. In this paper we use the Riccati equations only on time horizons for which a solution is known to exist. Proposition~\ref{prop:extended-Riccati-existence} provides local existence up to the maximal lifetime, while global existence is verified explicitly only in the concrete examples treated later.
\end{remark}
Let $\cH^{\MC}$ denote the complexification of $\cH$. We introduce the following notation
 \begin{align*}
  S(\cO)\df\set{u\in\cH^{\MC}\colon\, \Re(u)\in\cO}.
\end{align*}
Note that for every $u\in S(\cO)$ we have
\begin{align*}
\int_{\cHplus\cap \set{\norm{\xi}_2>1}}|\E^{-\langle
  \xi,u\rangle}|M(y,\D\xi)=\int_{\cHplus\cap
  \set{\norm{\xi}_2>1}}\E^{-\langle \xi,\Re(u)\rangle}M(y,\D\xi)<\infty\,, \quad \forall y\in \cHplus\,.
\end{align*}
\begin{definition}
We shall consider $\mathsf F$
and $\mathsf R$ as functions from $S(\cO)$ to $\MC$ resp.~from $H^{\MC} \times S(\cO)$ to $\cH^{\MC}$. We then
consider the following \emph{complex generalized Riccati equations}:
 \begin{subequations}
\begin{empheq}[left=\empheqlbrace]{align} 
  \,\frac{\partial {\mathsf \phi}}{\partial t}(t,u)&= \mathsf F(\psi_{2}(t,u)),\quad 0<t\leq T, \quad 
  \phi(0,u)=0, \label{eq:complex-Riccati-phi}\\
  \,\psi_{1}(t,u)&= u_{1}+\cA^{*}\int_{0}^{t}\psi_{1}(s,u)\D s, \quad 0<t\leq T, \quad \psi_{1}(0,u)=u_{1}\,,  \label{eq:complex-Riccati-psi-1}\\
  \,\frac{\partial \psi_{2}}{\partial t}(t,u)&=\sR(\psi_{1}(t,u),\psi_{2}(t,u)), \quad 0<t\leq T,\quad \psi_{2}(0,u)=0. \label{eq:complex-Riccati-psi-2}
\end{empheq}
\end{subequations}
Let $T\geq 0$, $u=(u_{1},0)$, for $u_1 \in H^{\MC}$. Analogously to the
extended generalized Riccati
equations~\eqref{eq:extended-Riccati-P}-\eqref{eq:extended-Riccati-q2}, we say
that the map $(\phi(\cdot, u), \psi_1(\cdot,u), \psi_2(\cdot, u))$ is a
solution to equations~\eqref{eq:complex-Riccati-phi}-\eqref{eq:complex-Riccati-psi-2} whenever $\phi(\cdot,u)\in C^{1}([0,T], \MC),\psi_{1}(\cdot,u)\in C([0,T],H^{\MC}),\psi_{2}(\cdot,u)\in
C^{1}([0,T],S(\cO))$ satisfies \eqref{eq:complex-Riccati-phi}-\eqref{eq:complex-Riccati-psi-2}.
\end{definition}
\subsubsection{Fourier-Laplace transforms of the pure-jump affine volatility model}
From Proposition~\ref{prop:existence-complex-solution}
and following similar derivations as in \cite[Theorem 2.26 and Section 5.3]{KM15}, we deduce the following result:
\begin{theorem}\label{prop:complex-extension}
Let $(b,B,m,\mu)$ be an admissible parameter set satisfying Assumptions~\ref{def:admissibility} and~\ref{ass:chkk-finite-variation},
let $(X,\sf Y)$ be the stochastic volatility model \eqref{eq:joint-affine-SDE}, and let $u_1 \in H^\mathbb{C}$.
Let $(\sf P, \sf q_1, \sf q_2)$ be the mild solution to the extended Riccati equations \eqref{eq:extended-Riccati-P}-\eqref{eq:extended-Riccati-q2} on $[0,T]$ with initial value $(\Re(u_1), 0)$, the existence of which is guaranteed by Proposition \ref{prop:extended-Riccati-existence}. Then 
it holds 
 $\EX{|\E^{\langle X_{t}, u_{1}\rangle_{H^{\mathbb{C}}}}|}<\infty$, and
 \begin{align*}
   \EX{\E^{\langle X_{T},u_{1}\rangle_{H^{\mathbb{C}}}}\mid \mathcal{F}_t}=\E^{-\phi(T-t,u_{1},0)+\langle
   X_t,\psi_{1}(T-t,u_{1},0)\rangle_{H^\mathbb{C}}-\langle Y_t, \psi_{2}(T-t,u_{1},0)\rangle_{\cH^\mathbb{C}}},  
 \end{align*}
 where $(\phi(\cdot,u_1,0), \psi_1(\cdot,u_1,0), \psi_2(\cdot,u_1,0))$ is the unique solution 
to~\eqref{eq:complex-Riccati-phi}-\eqref{eq:complex-Riccati-psi-2} on $[0,T]$ the existence of which is guaranteed by Proposition \ref{prop:existence-complex-solution}.
\end{theorem}

\subsection{The infinite-dimensional Wishart stochastic
volatility model}\label{sec:OU-Wishart}

As in the case of the pure-jump stochastic volatility model, we start by introducing the instantaneous covariance process that defines a Gaussian infinite-dimensional stochastic volatility model.

\subsubsection{Infinite-dimensional Wishart process}
Let $\mathbb{A} \in \cL(H)$, let $Q \in \cS^+(H)$, let $n \in \mathbb{N}$, and let $y \in \cS_1^+(H)$ be of rank at most $n$. Moreover, assume that  
$$\int_0^t \|\E^{s\mathbb{A}}\sqrt{Q}\|^2_{2} \, \D s< \infty\,,$$
for all $t\geq 0$. Consider the following stochastic differential equation 
\begin{align}\label{eq:infinite-dimensional-wishart}
\D \cY_t &= n Q\, \D t+ \cY_t \mathbb{A} \D t + \mathbb{A}^* \cY_t\, \D t +\sqrt{\cY}_t\, \D B_t\sqrt{Q} + \sqrt{Q} \D B_t^*\sqrt{\cY}_t\,, \quad t\geq 0\,,\nonumber\\
\cY_0 &= y\,,
\end{align}
where $(B_t)_{t\geq0}$ is an $\cL_2(H)$-cylindrical Brownian motion.
As shown in \cite[Theorem 2.1]{cox2024infinite}, equation~\eqref{eq:infinite-dimensional-wishart} admits a probabilistically and analytically weak solution. In fact, the existence result in \cite[Theorem 2.1]{cox2024infinite} holds more generally when $\mathbb{A}$ is an unbounded operator generating a strongly-continuous semigroup.

Furthermore, \cite[Proposition 4.7 and Corollary 4.11]{cox2024infinite} provide a characterization of finite-rank $\cS_1^+(H)$-valued Wishart processes. In this paper, however, we restrict our attention to the case of a bounded operator $\mathbb{A}$, since under this assumption equation~\eqref{eq:infinite-dimensional-wishart} can naturally serve as a model for the instantaneous covariance process as we argue next.

\subsubsection{The infinite-dimensional Wishart stochastic volatility model}

Let $n \in \mathbb{N}$, let $y \in \cS_1^+(H)$ be of rank at most $n$, and let $(W_t)_{t \geq 0}$ be an $H$-cylindrical Brownian motion.  
Let $x, \Upsilon \in H$ and $D \in \cS_1^+(H)$.

In this second model, the instantaneous covariance process is given by the square root of the Wishart process~\eqref{eq:infinite-dimensional-wishart}.  
For parameters $(\mathbb{A}, Q, D, \Upsilon, \cA)$ and initial value $(x, y)$, we consider the coupled SDE
\begin{align}\label{eq:stoc-vol-1}
\begin{cases}
\D X_{t} = \big( \cA X_{t} + D^{1/2} \cY_t D^{1/2} \Upsilon \big) \,\D t
           + D^{1/2} \cY_t^{1/2} \,\D W_{t}, \\[4pt]
\D \cY_t = \big( n Q + \cY_t \mathbb{A}+ \mathbb{A}^* \cY_t \big) \,\D t
           + \sqrt{\cY_t} \,\D B_t \sqrt{Q}
           + \sqrt{Q} \,\D B_t^* \sqrt{\cY_t}, 
           \quad t \geq 0 ,
\end{cases}
\end{align}
where $(W_t)_{t \geq 0}$ and $(B_t)_{t \geq 0}$ are independent.  
The existence of a (stochastically) weak solution to~\eqref{eq:stoc-vol-1} in $H \times \cS_1^+(H)$ follows from \cite[Lemma~2.8]{CKK22b} and \cite[Theorem~2.1]{cox2024infinite}. Similarly to the pure-jump stochastic volatility model, we include the linear term $D^{1/2}\cY_t D^{1/2}\Upsilon$ in the drift of $(X_t)_{t\geq0}$, as motivated by the no-arbitrage drift condition; see Section~\ref{sec:appl-math-finance}.

We refer to~\eqref{eq:stoc-vol-1} as the \emph{infinite-dimensional Wishart stochastic volatility model}.

\subsubsection{Riccati equations associated with the infinite-dimensional Wishart stochastic volatility model} 

Let $n \in \mathbb{N}$, $Q, D \in \cS_1^+(H)$, $\Upsilon \in H$, and $\mathbb{A} \in \cL(H)$.  
Define $F \colon \cL_1(H) \to \mathbb{R}$ and $R \colon H \times \cL_1(H) \to \cL_1(H)$ by
\begin{align}
F(u) &= n \, \tr(Q u), \label{eq:F-Wishart}\\
R(h,u) &= u \mathbb{A}^* + \mathbb{A} u - \tfrac{1}{2} \big(D^{1/2} h\big)^{\otimes 2}
          - D^{1/2} h \otimes D^{1/2} \Upsilon \nonumber \\
       &\quad - \tfrac{1}{2} \big(u + u^T \big) Q \big(u + u^T \big). \label{eq:R-Wishart}
\end{align}

As in the affine pure-jump case, the functions $F$ and $R$ give rise to a system of operator-valued Riccati-type equations.  
Because of the mixed term $D^{1/2}h\otimes D^{1/2}\Upsilon$, the Riccati vector field $R(h,\cdot)$ does not preserve self-adjointness in general. Accordingly, the Riccati variable is matrix-valued, even though the Wishart state process $(\cY_t)_{t\ge0}$ itself remains in the cone of self-adjoint positive operators.
Since the operators $(\mathcal{Y}_t)_{t\geq0}$ of the Wishart process are finite rank in the present setup once the initial condition has rank at most $n$, we adopt a numerically convenient simplification:  
we approximate the model by a matrix-valued Wishart process whose rank remains constant over time.  
This reduces the infinite-dimensional problem in the covariance component to a more tractable finite-dimensional one while preserving the maturity-space structure of the forward curve model.

Let $(H_n)_{n \in \mathbb{N}}$ denote an increasing sequence of finite-dimensional subspaces such that $\overline{\bigcup_{n} H_n} = H$. Let $\{e_1, \ldots, e_n\}$ be an orthonormal basis of $H_n$ and define the orthogonal projection $\Pi_n: H \to H_n$. We denote by $\mathcal{L}_n := \mathcal{L}(H_n) \cong \mathbb{R}^{n \times n}$ the space of bounded linear operators on $H_n$ and by $\mathcal{S}_n$ the space of real symmetric $n \times n$ matrices. Define the finite-rank counterparts $F_n \colon \mathcal{L}_n \to \mathbb{R}$ and $R_n \colon H \times \mathcal{L}_n \to \mathcal{L}_n$ of the nonlinear maps $F$ and $R$ introduced in \eqref{eq:F-Wishart} and \eqref{eq:R-Wishart} via
\begin{equation}\label{Fn-Rn}
F_n(X) := F(\iota_n X \iota_n^*), \qquad R_n(y, X) := \Pi_n R(y, \iota_n X \iota_n^*) \Pi_n^*,
\end{equation}
where $\iota_n: \mathbb{R}^n \hookrightarrow H$ is the canonical injection.
\begin{definition}\label{def:finite-rank-riccati}
Let $F_n$ and $R_n$ be as in \eqref{Fn-Rn}. Then, for $u = (u_1, u_2) \in H \times \mathcal{L}_n$ and $T \geq 0$, the \emph{finite-rank Riccati system} is given by the system of ordinary differential equations
\begin{subequations}\label{eq:riccati-finite-Wishart}
\begin{empheq}[left=\empheqlbrace]{align} 
  \frac{\partial P_n}{\partial t}(t,u) &= F_n(q_{2,n}(t,u)), \quad && P_n(0,u) = 0, \label{eq:riccati-finite-phi}\\
  q_{1}(t,u) &= u_1 + \cA^* \int_0^t q_{1}(s,u) \, \D s, \quad && q_{1}(0,u) = u_1, \label{eq:riccati-finite-psi-1}\\
  \frac{\partial q_{2,n}}{\partial t}(t,u) &= R_n(q_{1}(t,u), q_{2,n}(t,u)), \quad && q_{2,n}(0,u) = u_2\,. \label{eq:riccati-finite-psi-2}
\end{empheq}
\end{subequations}
We call $(P_n, q_{1}, q_{2,n})$ a \emph{solution} to \eqref{eq:riccati-finite-Wishart} if $P_n \in C^1([0,T], \mathbb{R})$, $q_{1} \in C^1([0,T], H)$, and $q_{2,n} \in C^1([0,T], \mathcal{L}_n)$ solve \eqref{eq:riccati-finite-phi}--\eqref{eq:riccati-finite-psi-2}.
\end{definition}

\begin{definition}\label{def:complex-riccati-Wishart}
Denote by $\cL_n^{\mathbb{C}}= \cL_n \oplus \mathrm{i}\cL_n$. Let $F_n$ be the analytic extension to $\cL_n^{\mathbb{C}}$ of the function $F_n$ defined in \eqref{Fn-Rn} and 
$R_n$ be the analytic extension to $H^\MC \times \cL_n^{\mathbb{C}}$ of the function $R_n$ defined in \eqref{Fn-Rn}. Consider the following system of ordinary differential equations
\begin{subequations}\label{eq:riccati-Wishart2}
\begin{empheq}[left=\empheqlbrace]{align} 
  \,\frac{\partial {\eta_n}}{\partial t}(t,u)&=F_n(\zeta_{2,n}(t,u)),\quad 0<t\leq T, \quad 
  \eta_n(0,u)=0, \label{eq:complex-Riccati-phi-Wishart}\\
  \,\zeta_{1}(t,u)&= u_{1}+\cA^{*}\int_{0}^{t}\zeta_{1}(s,u)\D s, \quad 0<t\leq T, \quad \zeta_{1}(0,u)=u_{1}\,,  \label{eq:complex-Riccati-psi-1-Wishart}\\
  \,\frac{\partial \zeta_{2,n}}{\partial t}(t,u)&=R_{n}(\zeta_{1}(t,u),\zeta_{2,n}(t,u)), \quad 0<t\leq T,\quad \zeta_{2,n}(0,u)=u_2. \label{eq:complex-Riccati-psi-2-Wishart}
\end{empheq}
\end{subequations}
For $u=(u_{1},u_{2})\in H^\mathbb{C}\times \cL_n^\MC$ and $T\geq 0$, we say that \sloppy $(\eta_n(\cdot, u), \zeta_{1}(\cdot,u), \zeta_{2,n}(\cdot, u)) \colon  [0,T] \rightarrow \mathbb{R}\times H^\mathbb{C}\times \cL_n^\MC$ is a mild solution to \eqref{eq:complex-Riccati-phi-Wishart}-\eqref{eq:complex-Riccati-psi-2-Wishart} if 
    $\eta_n(\cdot,u)\in C^{1}([0,T],\MR)$, $\zeta_{1}(\cdot,u)\in C^1([0,T],H^\mathbb{C})$ and $\zeta_{2,n}(\cdot,u)\in
    C^{1}([0,T], \cL_n^{\mathbb{C}})$ satisfies \eqref{eq:complex-Riccati-phi-Wishart}-\eqref{eq:complex-Riccati-psi-2-Wishart}.
    Then we call $(\eta_n,\zeta_1,\zeta_{2,n})$ a solution (up to time $T$ and with starting point $u$) of \emph{the complex Riccati} system.
\end{definition}

Next, we state an existence result for the complex Riccati system \eqref{eq:riccati-Wishart2}. 
The proof follows the same lines as Proposition~\ref{prop:existence-complex-solution}; 
see also \cite[Proposition~5.1]{Kel10} for the case $D=S_d^+$, $d\geq 2$.

\begin{theorem}\label{thm:extended-affine-Z}
Let $T>0$ and $n\in\mathbb{N}$. Let $Q,D\in\cS_1^+(H)$, $\Upsilon\in H$, and $\mathbb{A}\in\cL(H)$. 
Let $(\cA,\dom(\cA))$ be the generator of a strongly continuous semigroup $(S_t)_{t\ge0}$ on $H$, 
and $(\cA^*,\dom(\cA^*))$ its adjoint. 
Let $u=(u_1,u_2)\in H^\mathbb{C}\times\cL_n^{\mathbb{C}}$ and suppose that the \emph{real} Riccati system
\eqref{eq:riccati-finite-phi}--\eqref{eq:riccati-finite-psi-2} with initial data $\Re(u)=(\Re(u_1),\Re(u_2))$ 
admits a (unique) mild solution $(P_n,q_{1,n},q_{2,n})$ on $[0,T]$, as guaranteed by 
Proposition~\ref{prop:extended-Riccati-existence-Wishart}. Then the \emph{complex} Riccati system 
\eqref{eq:complex-Riccati-phi-Wishart}--\eqref{eq:complex-Riccati-psi-2-Wishart} with initial data $u$ 
admits a (necessarily unique) mild solution $(\eta_n,\zeta_{1,n},\zeta_{2,n})$ on $[0,T]$.
\end{theorem}

\subsubsection{Approximation of the Laplace transform of the infinite-dimensional Wishart stochastic volatility model} 

Our primary focus in this paper is the numerical implementation of option pricing in a setting where the forward price or rate follows an infinite-dimensional stochastic differential equation. To enable computations, we approximate the model by finite-dimensional systems.  

For the Wishart stochastic volatility model, we do not perform a detailed convergence analysis of these approximations. Such an analysis, while important for theoretical completeness, would introduce significant technical complexity and is therefore left for future work.  
Our objective here is to demonstrate the practical feasibility and computational tractability of the method. For comparison, the case in which the covariance process is driven by jumps has been analyzed in~\cite{Kar24,Kar25}.

To keep the presentation concise and highlight the numerical methodology, we approximate the Laplace transform of the infinite-dimensional Wishart stochastic volatility model directly.  
Let $(X_t, \cY_t)_{t \geq 0}$ be as in~\eqref{eq:stoc-vol-1}, and let $(\eta_n, \zeta_1, \zeta_{2,n})$ be the solution to~\eqref{eq:riccati-Wishart2} with initial condition $(u_1, u_2) \in H^{\mathbb{C}} \times \cL_n^{\mathbb{C}}$, whose existence is guaranteed by Theorem~\ref{thm:extended-affine-Z}.  

In our option pricing derivations in the next section, we use the following finite-rank approximation to the Laplace transform:
\begin{align}\label{eq:Wishart-price-approx}
&\mathbb{E}\left[
    \E^{\langle X_{T}, u_{1} \rangle_{H^\mathbb{C}} - \tr\big( \cY_{T} \, \iota_n u_{2} \iota_n^* \big)}
    \,\middle|\, \mathcal{F}_t
\right] \nonumber\\
&\qquad\approx
\exp\Big(
    - \eta_n(T-t, u)
    + \langle X_t, \zeta_{1}(T-t, u) \rangle_{H}
    - \tr\big( \cY_t \, \iota_n \zeta_{2,n}(T-t, u) \iota_n^* \big)
\Big).
\end{align}
The approximation in~\eqref{eq:Wishart-price-approx} is the object used in the numerical section below. For general infinite-dimensional data we do not prove convergence of the right-hand side as $n\to\infty$. 

\section{Pricing options on forwards in infinite-dimensional stochastic volatility models}\label{sec:appl-math-finance}

We begin with a brief introduction to the modeling of commodity forward
curves in the Heath-Jarrow-Morton framework, here we closely follow
\cite{BK14} and connect our two infinite-dimensional stochastic volatility models $(X_t, \cY_t)_{t\geq 0}$ and $(X_t, {\sf Y}_t)_{t\geq 0}$ to this framework.\par

Suppose $w\colon \R^+ \rightarrow  [1, \infty)$ is a non-decreasing function such that $w(0) = 1$ and $\int_{0}^{\infty}w^{-1}(x)\,\D
 x<\infty$. Let $H_w$ be the space of
absolutely continuous functions $f: \R^+\rightarrow \R$ such that  
 \begin{align*}
  \norm{f}_{w}^{2}:=f(0)^{2}+\int_{0}^{\infty}w(x)|f'(x)|^{2}\,\D x < \infty\,,
\end{align*}
where $f'$ is the weak derivative of $f$. The space $H_w$ is a separable Hilbert space when endowed with the inner product 
\begin{align*}
 \langle f,g\rangle_w := f(0)g(0) +\int_{0}^{\infty}w(x)f'(x)g'(x)\,\D x.
\end{align*}
This space is proposed in \cite{Fil01} as a class of Hilbert spaces that match the economic reasoning about the forward rate curve in fixed income markets (see also~\cite{BK14} for the case of commodity forward markets).

The left-shift semigroup is a strongly continuous semigroup on $H_w$ with infinitesimal generator $\cA$ being the operator of differentiation (in time-to-maturity coordinate). 
Moreover, the point evaluation $\delta_x(h) = h(x)$ is a continuous linear functional on $H_{w}$, for all $x \in \R^+$. 
It can be expressed as (see \cite[Lemma 5.3.1]{Fil01})
 \begin{align}\label{eq:insert-element}
 \delta_{x}=\langle \cdot, u_{x}\rangle_{w}\,,  
 \end{align}
 where for $x \in \R^+$, $u_x: \R^+ \rightarrow \R$, $y \mapsto 1+\int_{0}^{x\wedge y}w^{-1}(z)\,\D z$.\par
In particular, setting $h_0:=u_0\equiv 1\in H_w$, we have $u_x=S^*(x)h_0$ for every $x\ge0$.

Let $F(t,\tau)$ be the forward price at time $t>0$ for delivery at $\tau>t$.
The log–forward in the Musiela parametrisation (time to maturity $x=\tau-t$) is
$$
  f_t(x) := \ln F(t,t+x), \qquad x \in \R_+.
$$
Fix a weight function $w$ and assume the state process $(X_t)_{t\ge0}$ satisfies
\begin{align}\label{eq:X-under-P}
  \D X_t
  = \big( \cA X_t + C_t \big)\,\D t
    + D^{1/2} Y_t^{1/2}\, \D W_t, 
  \qquad X_0 = x \in H_w,
\end{align}
where $\cA=\partial_x$ is the derivative (shift) operator, $C$ is an integrable adapted $H_w$–valued process, 
$D\in\cL(H_w)$ is self-adjoint and positive, $(W_t)_{t\ge0}$ is an $H_w$-cylindrical Brownian motion, and $(Y_t)_{t\ge0}$
is either $(\cY_t)_{t\ge0}$ from~\eqref{eq:infinite-dimensional-wishart} or $({\sf Y}_t)_{t\ge0}$ from~\eqref{eq:process-sY}.
On a filtered probability space $(\Omega,\cF,(\cF_t)_{0\le t\le T},\MP)$, we model the (arbitrage-free) forward curve via
\begin{equation}\label{dynamics-f}
  F(t,t+x) := \exp\{\delta_x(X_t)\}
             = \exp\{\langle X_t, u_x\rangle_w\}, 
  \qquad 0\le t\le T,
\end{equation}
where $u_x$ is the Riesz representer of $\delta_x$ (see~\eqref{eq:insert-element}). Assume $\xi$ is an $H_w$–valued adapted process such that
\begin{align}\label{eq:girsanov-condition}
  \mathbb{E}\left[\exp\Big\{-\int_0^T \langle \xi_s, \D W_s\rangle_w
                 - \tfrac12 \int_0^T \|\xi_s\|_w^2\,\D s\Big\}\right]=1.
\end{align}
By Girsanov’s theorem for cylindrical Brownian motion (e.g., \cite[Thm.~13]{da2013strong}), 
$$
  \widetilde W_t := W_t - \int_0^t \xi_s\,\D s
$$
is a cylindrical Brownian motion under $\MQ$ with density
$$
  \left.\frac{\D\MQ}{\D\MP}\right|_{\cF_t}
  = \exp\Big\{-\int_0^t \langle \xi_s, \D W_s\rangle_w
                 - \tfrac12 \int_0^t \|\xi_s\|_w^2\,\D s\Big\}, \qquad 0\le t\le T.
$$
Under $\MQ$, the process $(X_{t})_{t\geq 0}$ solves
\begin{equation}\label{eq:X-under-Q}
  \D X_t
  = \big(\cA X_t + C_t - D^{1/2} Y_t^{1/2}\xi_t\big)\,\D t
    + D^{1/2} Y_t^{1/2}\,\D \widetilde W_t, \qquad 0\le t\le T.
\end{equation}
The following result gives a drift condition ensuring the martingale property of $t \mapsto F(t,\tau)$, for $t\leq \tau$. The proof follows the same lines as \cite[Proposition 6.1]{zbMATH07927020}
\begin{lemma}\label{lem:drift-condition}
Let $\xi$ satisfy \eqref{eq:girsanov-condition}. Let $\cA$, $C$, $D$, and $(Y_t)_{t\ge0}$ be as above, and 
let $S(t)$ be the left–shift semigroup generated by $\cA$. Recall that $u_x=S^*(x)h_0$ with $h_0\equiv 1$. Suppose that, for all $0\le t\le \tau\le T$,
\begin{equation}\label{eq:drift-condition}
  C_t - D^{1/2} Y_t^{1/2}\xi_t
  = -\tfrac12\, D^{1/2} Y_t D^{1/2} \, u_{\tau-t}.
\end{equation}
Then, for each fixed $\tau$, the process $t\mapsto F(t,\tau)=\exp\{\langle X_t,u_{\tau-t}\rangle_w\}$ is a (local) $\MQ$–martingale on $[0,\tau]$.
\end{lemma}

From now on we impose the no–arbitrage drift condition \eqref{eq:drift-condition} for a fixed delivery lag $\vartheta$. 
Equivalently, in the coupled SV dynamics \eqref{eq:stoc-vol-1} or \eqref{eq:joint-affine-SDE} used below to price claims with delivery lag $\vartheta$ under $\MQ$ we choose
\begin{equation}\label{eq:upsilon}
  \Upsilon^\vartheta := -\tfrac12\, u_\vartheta
  = -\tfrac12\, S^*(\vartheta)h_0, 
  \qquad \vartheta\ge0,
\end{equation}
Indeed, with this parametrization the model drift satisfies 
$C_t - D^{1/2}Y_t^{1/2}\xi_t = D^{1/2}Y_t D^{1/2}\Upsilon^\vartheta$,
and for a claim with delivery lag $\vartheta=\tau-t$ this equals 
$-\tfrac12 D^{1/2}Y_t D^{1/2} S^*(\tau-t)h_0$, which is exactly \eqref{eq:drift-condition}.
Moreover, $\vartheta\mapsto D^{1/2}S^*(\vartheta)h_0 \in H_w$ (see \cite[Lem.~6.2]{zbMATH07927020}), 
so $\Upsilon^\vartheta\in H_w$ is well defined.
Accordingly, all pricing results below are formulated under a risk-neutral specification for a fixed delivery lag $\vartheta$. This is sufficient for pricing the corresponding family of forward options, but we do not claim a single maturity-uniform HJMM drift specification for all delivery lags simultaneously.

In pricing, all quantities depend on the delivery lag $\vartheta:=T_1-T_0$ through the evaluation functional $u_\vartheta=S^*(\vartheta)h_0$
and the choice $\Upsilon^\vartheta$ in \eqref{eq:upsilon}. We therefore write the Riccati solutions as
$(\eta^\vartheta,\zeta^\vartheta_1,\zeta^\vartheta_2)$ for the Wishart specification and
$(\phi^\vartheta,\psi^\vartheta_1,\psi^\vartheta_2)$ for the pure-jump specification. 
For an option written at time $t$ on the forward delivering at $T_1$ and set at $T_0$ we have
\begin{equation}\label{eq:forward-rates}
  F(T_0,T_1) 
  = \exp\{\langle X_{T_0}, u_{T_1-T_0}\rangle_w\}
  =: F_{T_0}^{(\vartheta)}, 
  \qquad 0\le T_0\le T_1\le T,
\end{equation}
a standard payoff in commodity markets such as oil, metals, and agriculture, see \cite{benth2012modeling}. 
In the next proposition, we derive semi-explicit call prices via the Fourier transform of the payoff and the
corresponding extended Riccati equations for our stochastic volatility models.

\begin{proposition}[Call option pricing in the pure-jump affine SV model]\label{prop:price-option}
Let $T>0$, $0\le T_{0}\le T_{1}\le T$, and set $\vartheta := T_{1}-T_{0}$.
Assume $(X^{\vartheta}_t)_{t\ge 0}$ is as in \eqref{eq:X-under-P}, with $\Upsilon = -\tfrac{1}{2}u_\vartheta$,
and let the forward price $F$ be defined by \eqref{eq:forward-rates}.
Write $\mathbb{E}_{\mathbb{Q}}$ for expectation under $\mathbb{Q}$, and define $g:\mathbb{R}\to\mathbb{C}$ by
\begin{equation}\label{eq:chkk-fourrier-payoff}
  g(\lambda) := \frac{K^{-(\nu -1+\mathrm{i}\lambda)}}{(\nu +\mathrm{i}\lambda)(\nu-1 +\mathrm{i}\lambda)}\,,
  \qquad \nu>1,\ \lambda\in\mathbb{R}.
\end{equation}
Set $u := (\nu + \mathrm{i}\lambda)\,u_{\vartheta}$. Assume the covariance process is $Y_t=\mathsf{Y}_t$ as in \eqref{eq:process-sY}, and let $(\sf P, \sf q_1, \sf q_2)$ be the mild solution to the extended Riccati equations \eqref{eq:extended-Riccati-P}--\eqref{eq:extended-Riccati-q2} on $[0,T]$ with initial value $(\Re(u), 0)$, the existence of which is guaranteed by Proposition~\ref{prop:extended-Riccati-existence}.
Assume in addition that the Fourier integrand in \eqref{eq:chkk-call-option-price} belongs to $L^1(\mathbb R,\D\lambda)$ almost surely.
Then for $0\leq t\leq T_0$ the time-$t$ price of the call with maturity $T_0$ and strike $K>0$ written on $F$ is
\begin{align}\label{eq:chkk-call-option-price}
\mathbb{E}_{\mathbb{Q}}\big[(F(T_0,T_1)-K)^+ \mid \mathcal{F}_t\big]
= \frac{1}{2\pi}\int_{\mathbb{R}}
g(\lambda)\,
\exp\Big(
& -\phi^{\vartheta}(T_0-t,u,0)
+ \langle \psi_{1}^{\vartheta}(T_0-t,u,0), X_{t}^\vartheta\rangle_{H_w^{\mathbb{C}}} \nonumber\\
&\hspace{1.6cm}
- \langle \psi_{2}^{\vartheta}(T_0-t,u,0), \mathsf{Y}_{t}\rangle_{\mathcal{H}^{\mathbb{C}}}
\Big)\,\D\lambda\,,
\end{align}
where $(\phi^{\vartheta},\psi_1^{\vartheta},\psi_2^{\vartheta})$ denotes the solution to the generalized Riccati equations~\eqref{eq:complex-Riccati-phi}--\eqref{eq:complex-Riccati-psi-2} on $[0,T]$ with initial data $(u,0)$.
\end{proposition}
\begin{proof}
Standard Fourier techniques for option pricing (see, e.g., \cite{eberlein2010analysis} or \cite[Theorem 10.6]{filipovic2009term}) yield
\begin{align*}
\mathbb{E}_{\mathbb{Q}}[(F(T_0,T_1) -K)^+|\mathcal{F}_t] = \frac{1}{2\pi} \int_{\R}  g(\lambda)
\mathbb{E}_{\mathbb{Q}}[\E^{\langle X^\vartheta_{T_0}, (\nu+\mathrm{i}\lambda)u_{\vartheta}\rangle_{H_w^\mathbb{C}}}|\mathcal{F}_t] \, \D \lambda\,.
\end{align*}
Then the statement follows from Theorem~\ref{prop:complex-extension}.
\end{proof}

\begin{proposition}[Approximate call option pricing in the Wishart SV model]\label{prop:price-option-wishart}
Under the same hypotheses as Proposition~\ref{prop:price-option}, assume instead that the covariance process is $Y_t=\mathcal{Y}_t$ as in \eqref{eq:infinite-dimensional-wishart}, and let $(P_n,q_{1,n},q_{2,n})$ be the mild solution to the finite-rank Riccati equations
\eqref{eq:riccati-finite-phi}--\eqref{eq:riccati-finite-psi-2} on $[0,T]$ with initial value $(\Re(u),0)$, the existence of which is guaranteed by Proposition~\ref{prop:extended-Riccati-existence-Wishart}.
Assume in addition that the Fourier integrand in \eqref{eq:call-option-price-wishart} belongs to $L^1(\mathbb R,\D\lambda)$ almost surely. Then, for $0\leq t\leq T_0$, the finite-rank approximation of order $n$ to the call price is
\begin{align}\label{eq:call-option-price-wishart}
\mathbb{E}_{\mathbb{Q}}\big[(F(T_0,T_1)-K)^+ \mid \mathcal{F}_t\big]
\approx \frac{1}{2\pi}\int_{\mathbb{R}}
g(\lambda)\,
\exp\Big(
& -\eta_n^{\vartheta}(T_0-t,u,0)
+ \langle \zeta_{1}^{\vartheta}(T_0-t,u,0), X_{t}^\vartheta\rangle_{H_w^{\mathbb{C}}} \nonumber\\
&\hspace{1.6cm}
- \tr\big( \mathcal{Y}_{t}\iota_n \zeta_{2,n}^{\vartheta}(T_0-t,u,0)\iota_n^*\big)
\Big)\,\D\lambda\,,
\end{align}
where $(\eta^{\vartheta}_n,\zeta_1^{\vartheta},\zeta_{2,n}^{\vartheta})$ denotes the solution to the Riccati equations~\eqref{eq:complex-Riccati-phi-Wishart}--\eqref{eq:complex-Riccati-psi-2-Wishart} on $[0,T]$ with initial data $(u,0)$.
The approximation arises from the finite-rank projection~\eqref{eq:Wishart-price-approx} of the Wishart process.
\end{proposition}
\begin{proof}
The Fourier decomposition of the call payoff is identical to the proof of Proposition~\ref{prop:price-option}. Substituting the finite-rank approximation~\eqref{eq:Wishart-price-approx} for the conditional Laplace transform and applying Theorem~\ref{thm:extended-affine-Z} yields the result.
\end{proof}

Consider a put option with strike $K>0$ and exercise time $T_0$, i.e., the payoff is given by $(K-F(T_0,T_1))^+$.
Similarly to the case of a call option, the expressions
 in \eqref{eq:chkk-call-option-price} and \eqref{eq:call-option-price-wishart} yield the price of the corresponding put option on $(X_t^\vartheta, {\sf Y}_t)_{t\geq 0}$ and $(X_t^\vartheta, \cY_t)_{t\geq 0}$, respectively, by taking $u =(\nu +\mathrm{i}\lambda) u_{\vartheta}$ with $\nu <0$ and $\lambda \in \R$, provided the corresponding transform exists and the Fourier integrand is integrable. We use these damping conditions only as sufficient assumptions for Fourier inversion and do not attempt to characterize a maximal admissible strip in full generality.

\section{Numerical analysis}\label{sec:numerics}
In this section we study several specifications of the coupled process $(X_t,Y_t)_{t\geq 0}$ and compute prices of options written on forwards. Our primary goal is to verify, across models, that the pure-jump option values from the \emph{affine, semi-explicit} formulas~\eqref{eq:chkk-call-option-price} and the Wishart transform quantities underlying~\eqref{eq:call-option-price-wishart} are consistent with direct or conditional Monte Carlo benchmarks within sampling error, and to compare their computational cost. We start by describing the general settings of our model specifications.

\subsection{General settings}
In the implementation, we choose as a weight function $w(x) = \E^{\alpha x}$, for some variable $\alpha>0$. The initial forward rate curve is chosen using a Nelson-Siegel curve. 
We know that $H_w$ is isometrically isomorphic to $\R \times  L^2[0,\infty)$ via the isomorphism
\begin{equation}\label{eq:isomorphism}
T(h) = (h(0), h' w^{1/2})\,, \qquad T^{-1}(c,f) = c+\int_0^\cdot f(y)w^{-1/2}\, \D y.
\end{equation}
Assume $(e_n)_{n\in N}$ is an orthonormal basis of $\mathbb{R} \times L^2[0, \infty)$, then $(f_n) = (T^{-1}e_n)_{n\in N}$ is an orthonormal basis of $H_w$. We choose the orthonormal basis $(1, 0)$ and $(0, g_n)$ for $\mathbb{R} \times L^2[0, \infty)$, where $(g_n)_{n \in \mathbb{N}}$ are Laguerre polynomials 
multiplied by a weight function $\E^{-x/2}$ (see Appendix~\ref{sec:appendix1}). 
It follows that the orthogonal basis $(f_n)_{n\in \mathbb{N}}$ for $H_w$ is given by 
\begin{equation}\label{eq:f_basis}
\begin{aligned}
f_1 &= 1\,, \\
f_{n+1} &= \int_0^{\cdot}g_{n-1}(s)w^{-1/2}(s)\, \D s,\quad n\geq 1\,. 
\end{aligned}
\end{equation}
We consider the option value at time $t=0$. We recall the representation of the evaluation operators $\delta_x$. It holds for any linear operator $\cV \in \cL(H_w)$,
\begin{align*}
    \cV^*f(x) =\delta_x\cV^*f = \langle \cV^*f, u_x\rangle_w = \langle f, \cV u_x\rangle_w.
\end{align*}
We infer for $(S^*(t))_{t\geq 0}$ and the weight function $w(x) =  \E^{\alpha x}$,
\begin{align}\label{eq:semigroup-excplicit-form}
		S^*(t)h(x) &= h(0) + h(0)\int_0^{x\wedge t}w^{-1}(s)\,\D s + \int_t^{x\vee t} \frac{h^\prime(s-t)w(s-t)}{w(s)} \D s\nonumber\\
        & = h(0) + \frac{1}{\alpha}h(0)(1-\E^{-\alpha(x\wedge t)}) + \E^{-\alpha t}\mathbbm{1}_{x\geq t}(h(x-t)-h(0))\,.
\end{align}
Throughout the numerics we set $h_0\equiv 1$, so that $u_\vartheta=S^*(\vartheta)h_0$ is the evaluation functional at maturity lag $\vartheta$. Accordingly, we take
$$
\Upsilon^\vartheta=-\frac{1}{2}\mathcal{S}^*(\vartheta)h_0,
$$
so that the HJM drift in $(X_t)_{t\geq0}$ is $D^{1/2}\mathsf Y_t D^{1/2}\Upsilon^\vartheta$.
In the experiments on the stochastic volatility model with jumps, we consider two examples: the model driven by a L\'evy subordinator with $B(Y_t)=0$, and the so called Barndorff–Nielsen–Shephard (BNS) model, i.e., $B(\sY_t) = (\mathbb{C} \sY_t + \sY_t \mathbb{C}^*)$.

\subsection{Validation of the First Forward Moment and Option Prices in the Finite-Rank Wishart Model}

For the Wishart stochastic volatility model, we first validate the \emph{first forward moment} against its empirical counterpart obtained by Monte Carlo. In the code this corresponds to the special transform input $u_1=u_\vartheta$ and $u_2=0$, hence to the quantity $\EX{\exp(\langle X_T,u_\vartheta\rangle_w)}$. This is weaker than a full complex-MGF validation, but it checks the risk-neutral drift restriction and the finite-rank Riccati implementation in the simplest transform configuration relevant for pricing.

Throughout this experiment we work in the following setting: the coefficients $Q$, $\mathbb A$, $D$ and the initial condition $Y_0$ are all supported on $H_n:=\mathrm{span}\{f_1,\dots,f_n\}$. In particular, for this numerical specification the finite-rank transform is exact on $H_n$ rather than merely asymptotic. The theoretical first moment is therefore evaluated via the matrix Riccati system on $H_n$, while the empirical first moment is computed from Monte Carlo samples as the sample average of $\exp\{\langle X_T,u_\vartheta\rangle_w\}$.

We define the initial value $Y_0$ of the Wishart process as a rank-$n$ positive operator given by
$$
Y_0 = \sum_{k=1}^n \lambda_k f_k \otimes f_k,
$$
where $\lambda_k > 0$ are constant eigenvalues chosen for numerical stability and interpretation. The volatility-driving Wishart process $(Y_t)_{t \geq 0}$ is governed by the SDE
$$
\D Y_t = n Q\,\D t + Y_t \mathbb{A}\,\D t + \mathbb{A}^* Y_t\,\D t + \sqrt{Y_t}\,\D B_t \sqrt{Q} + \sqrt{Q}\,\D B_t^* \sqrt{Y_t},
$$
where $(B_t)_{t \geq 0}$ is a cylindrical Brownian motion on $L^2(H)$, and $Q \in \mathbb{S}_1^+(H)$ and $\mathbb{A} \in \mathcal{L}(H)$ are model parameters. We specify the operator $Q$ as
$$
Q = \sum_{k=1}^n q_k f_k \otimes f_k,
$$
with $q_k > 0$ to control the volatility of the Wishart process along each principal direction $f_k$. The drift operator $\mathbb{A}$ is chosen to be diagonal in the basis $(f_k)_{k\in\MN}$, such that
$$
\mathbb{A} = -\sum_{k=1}^n a_k f_k \otimes f_k,
$$
with $a_k > 0$ ensuring mean reversion of the covariance process. In this example, we set the $q_k=a_k=\frac{1}{k^2}$, $k\in \mathbb{N}$. We use $D \in \mathbb{S}_1^+(H)$ for the volatility loading, defined as
$$
D = \frac{1}{2} \sum_{k=1}^{\infty} \frac{1}{k^2} f_k \otimes f_k.
$$


The Wishart process is coupled with the forward dynamics via the volatility-modulated SPDE given in equation~\eqref{eq:HJMM}.

A solution to the Riccati equations~\eqref{eq:complex-Riccati-phi-Wishart}-\eqref{eq:complex-Riccati-psi-2-Wishart}, with $F$ and $R$ specified in~\eqref{eq:F-Wishart} and~\eqref{eq:R-Wishart}, is not available in closed form.  
In Appendix~\ref{sec:riccat_ws}, we present an efficient numerical scheme for solving these equations under the given parameterization.  
The method is flexible and can be adapted to alternative parameter choices.

For Monte Carlo simulation we exploit the standard factor representation of the finite-dimensional Wishart diffusion. Writing the matrix-valued covariance process on $H_n$ as
$$
Y_t = U_t^\top U_t,
$$
with $U_t\in\R^{n\times n}$ solving the Ornstein--Uhlenbeck SDE
$$
\D U_t = U_t A_n\,\D t + \D W_t\,Q_n^{1/2},
$$
one recovers exactly the finite-dimensional Wishart dynamics on $H_n$. In particular, the simulated covariance matrices remain symmetric positive semidefinite at the grid points without any diagonal reduction. For the first-moment experiment we use the short horizon $T_0 = 1/365$ and compare the empirical average of $\exp\{\langle X_{T_0},u_\vartheta\rangle_w\}$ with the affine/Riccati value.

In this corrected implementation, the theoretical first moment aligns closely with its Monte Carlo estimate, thereby confirming the consistency of the finite-rank affine transform in this benchmark configuration. This agreement is illustrated in Figure~\ref{fig:mmt_ws}.

\begin{figure}[ht]
    \centering
    \includegraphics[width=1\linewidth]{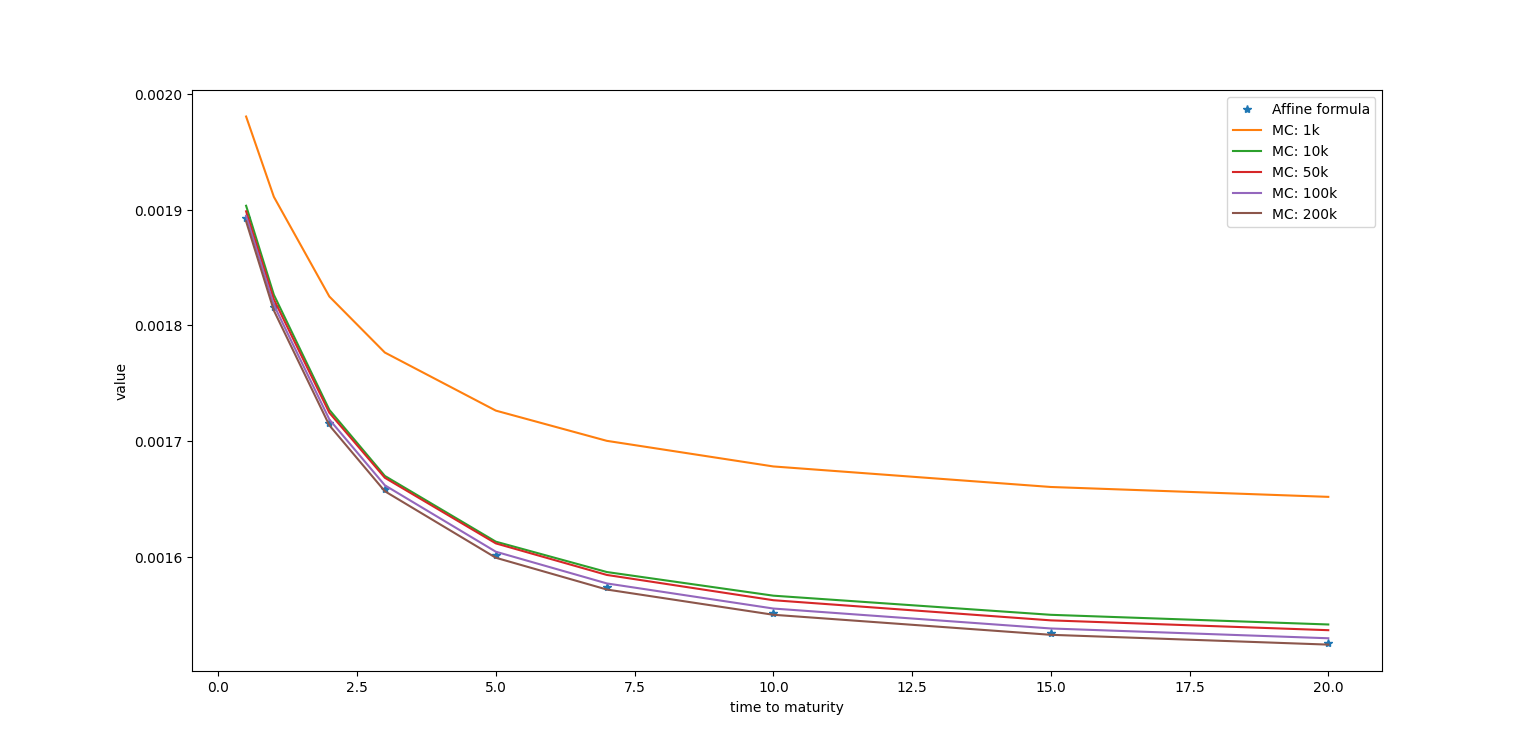}
    \caption{First forward moment $\EX{\exp(\langle X_{T_0},u_\vartheta\rangle_w)}$: Monte Carlo vs.\ affine formula, with $h_0=1, \alpha=0.1, T_0=1/365, N=5$.}
    \label{fig:mmt_ws}
\end{figure}

To address the option pricing validation, we employ a conditional Gaussian Monte Carlo scheme.
Given a simulated path of the \emph{full} finite-rank Wishart process $(Y_s)_{0\le s\le t}$, the log-forward $\log F(t,t+\vartheta)=\langle X_t,u_\vartheta\rangle_w$ is conditionally Gaussian with conditional mean
$$
  m_{\mathrm{cond}}
  = \langle X_0,u_{\vartheta+t}\rangle_w
    - \frac12\int_0^t
      \big\langle D^{1/2}Y_sD^{1/2}u_\vartheta,\,
      u_{\vartheta+t-s}\big\rangle_w\,\D s
$$
and conditional variance
$$
  \sigma^2_{\mathrm{cond}}
  = \int_0^t
    \big\langle D^{1/2}Y_sD^{1/2}u_{\vartheta+t-s},\,
    u_{\vartheta+t-s}\big\rangle_w\,\D s.
$$
Hence a call option price can be computed path-wise via the Black formula. In contrast to the previous diagonal-CIR surrogate, this conditional Gaussian scheme retains all off-diagonal covariance terms of $Y_s$ through the full quadratic forms above. Numerically, the former structural pricing gap disappears after this correction: in our implementation the remaining difference between conditional Monte Carlo and the affine formula is of the order expected from Monte Carlo noise together with time/Fourier discretization, whereas the diagonal surrogate systematically overprices by omitting the Wishart cross-correlation terms.

\begin{figure}[ht]
    \centering
    \includegraphics[width=1\linewidth]{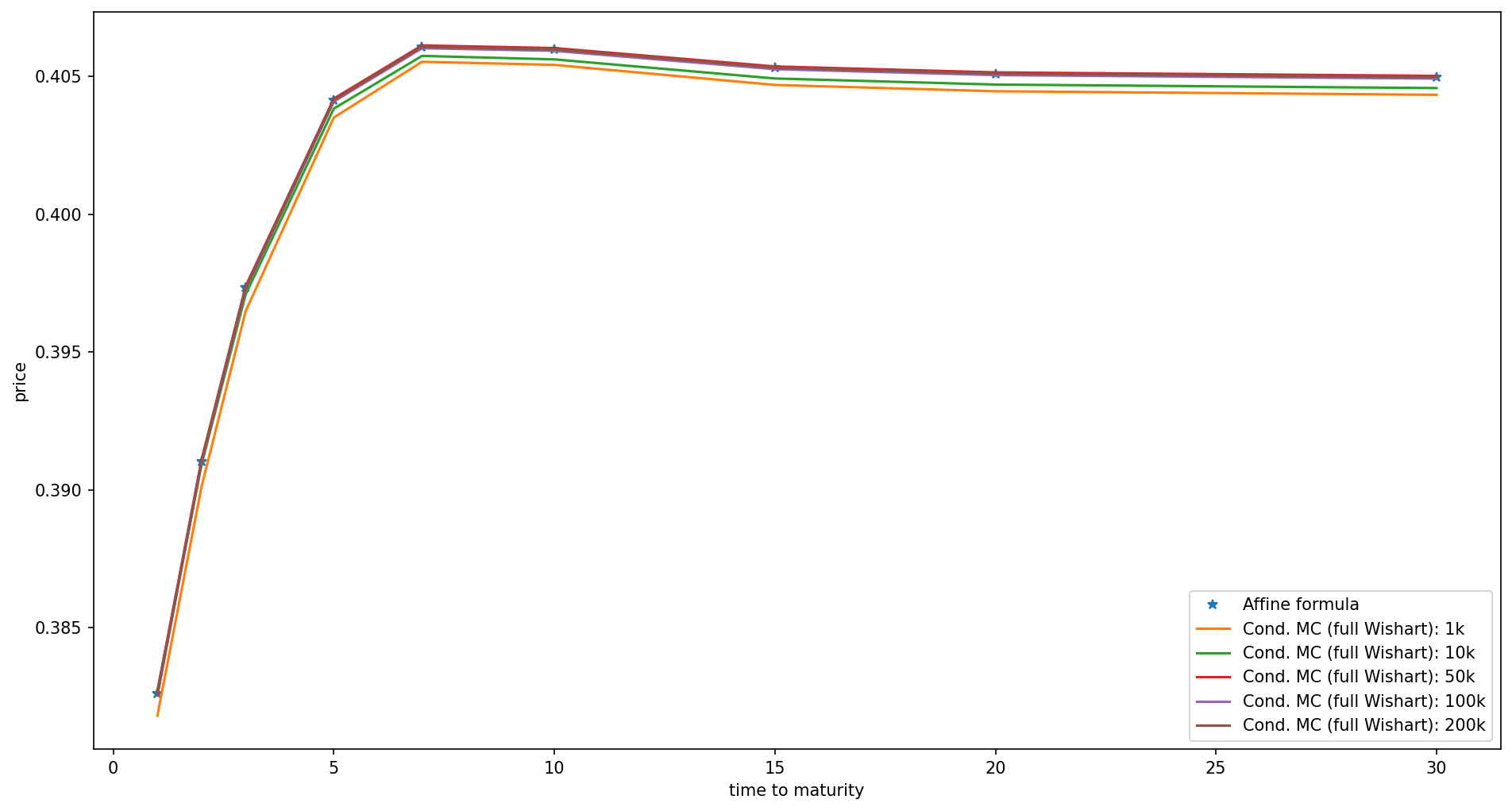}
    \caption{Call option prices: conditional Gaussian Monte Carlo based on the full finite-rank Wishart process vs.\ affine formula, with $h_0=1, \alpha=0.1, t=1, k=1, N=5$. }
    \label{fig:wishart_pricing}
\end{figure}

We emphasize that the slight decline of the pricing curves in Figure~\ref{fig:wishart_pricing} for large values of $\vartheta$ is not a numerical artefact. In this experiment the exercise time $t$ is fixed, and the horizontal axis varies the delivery lag $\vartheta$ in the forward contract $F(t,t+\vartheta)$. Hence the figure does not display option prices as a function of option maturity, and monotonicity in $\vartheta$ should not be expected. In the present specification, the deterministic initial forward level $F(0,t+\vartheta)$ decreases mildly with $\vartheta$, whereas the Wishart volatility loadings entering the conditional variance increase only up to a point and then essentially saturate. Consequently, for large $\vartheta$ the volatility contribution becomes nearly flat while the underlying forward level continues to decrease slightly, which leads to a small downward bend in the call prices. Any residual pointwise irregularity is due to Monte Carlo and Fourier discretization, but the overall long-end shape is a genuine feature of the model under this parameter choice.

\subsection{Pricing in a pure-jump L\'evy driven stochastic volatility model}\label{sec:example2}
As a second example we consider an affine pure-jump L\'evy–type stochastic volatility specification in which the
instantaneous covariance is an operator–valued L\'evy subordinator.

Let $(f_n)_{n\in\N}$ be the orthonormal basis of $H_w$ from \eqref{eq:f_basis}, and define the self-adjoint, positive and diagonal in $(f_n)_{n\in\MN}$ operator
$$
  D = \frac{1}{2}\sum_{n=1}^{\infty} \frac{1}{n^2}\, f_n\otimes f_n \in \cL(H_w).
$$
Let $\cA$ be the generator of the left-shift semigroup and $(\widetilde W_t)_{t\geq 0}$ a cylindrical Brownian motion under $\MQ$ as in Section~\ref{sec:appl-math-finance}.
Consider the compensated compound Poisson martingale $(J_t)_{t\geq 0}$ with deterministic jump directions $D$ and intensity $\beta>0$:
$$
  J_t = (N_t-\beta t)D, \qquad N_t \sim \mathrm{Poisson}(\beta t).
$$
Let $\cH_w$ denote the space of self-adjoint Hilbert–Schmidt operators on $H_w$ and $\cH_w^+$ its positive cone.
Under $\MQ$ we take the coupled model
\begin{align}\label{eq:sto-vol-pure-jump}
  \D X_t &= \big(\cA X_t + D^{1/2}\sY_t D^{1/2}\Upsilon^\vartheta\big)\,\D t
           + D^{1/2}\sY_t^{1/2}\,\D\widetilde W_t, 
           \qquad X_0=x\in H_w, \nonumber\\
  \D \sY_t &= \beta D\,\D t + \D J_t, 
           \qquad \sY_0=y\in \cH_w^+,
\end{align}
with $\Upsilon^\vartheta=-\tfrac12 S^*(\vartheta)h_0$ as in \eqref{eq:upsilon}.
The numerical experiments in Sections~\ref{sec:example2} and~\ref{sec:example3} specialize to the Lévy-driven case $\mu=0$ and to a compound-Poisson OU/BNS specification with jumps in the fixed direction $D$. Section~\ref{sec:example4} extends the validation to the fully state-dependent jump kernel from Section~\ref{sec:OU-jumps}.

The complex associated Riccati system \eqref{eq:complex-Riccati-phi}–\eqref{eq:complex-Riccati-psi-2} becomes
\begin{align*}
  \partial_t \phi(t,u)=F(\psi_2(t,u)),\qquad 
  \psi_1(t,u)=u_1+\cA^*\int_0^t\psi_1(s,u)\,\D s,\qquad
  \partial_t\psi_2(t,u)=R(\psi_1(t,u),\psi_2(t,u)),
\end{align*}
with
\begin{align*}
  R(h,u)= -\tfrac12\,(D^{1/2}h)^{\otimes 2} - D^{1/2}h\otimes D^{1/2}\Upsilon^\vartheta,
  \qquad
  F(u)= -\beta\big(\E^{-\langle D,u\rangle}-1\big).
\end{align*}
The following lemma states closed-form expressions of the solutions to the system in the basis $(f_n)$.

\begin{lemma}\label{lem:Riccati-pure-jump-stoc-vol}
Let $w(x)=\E^{\alpha x}$ with $\alpha>0$ and $(f_n)_{n\in\N}$ as above. For $u_1=(\nu+i\lambda)\,u_\vartheta$ with $\nu>1$ and $\lambda\in\R$, the Riccati solution satisfies
\begin{align}
  \psi_1(t,u)(x) 
  &= (\nu+i\lambda)\!\left(1+\frac{1}{\alpha}\big(1-\E^{-\alpha(x\wedge t)}\big)
                                +\frac{1}{\alpha}\E^{-\alpha t}\,\mathbf 1_{\{x\ge t\}}
                                   \big(1-\E^{-\alpha(\vartheta\wedge x-t)}\big)\right),
  \label{eq:riccati_pure_jump1}\\[4pt]
  \psi_2(t,u)
  &= -\frac{(\nu+i\lambda)^2}{4}
     \sum_{i=1}^{\infty}\sum_{j=1}^{\infty}\frac{1}{ij}\,
     \langle f_i,\cdot\rangle_w\, f_j \int_0^t f_i(s+\vartheta)f_j(s+\vartheta)\,\D s
     \nonumber\\[-2pt]
  &\quad\; -\frac{(\nu+i\lambda)}{2}
     \sum_{i=1}^{\infty}\sum_{j=1}^{\infty}\frac{1}{ij}\,c_j(\vartheta)\,
     \langle f_i,\cdot\rangle_w\, f_j \int_0^t f_i(s+\vartheta)\,\D s,
  \label{eq:riccati_pure_jump2}
\end{align}
where 
$D^{1/2}=\tfrac{1}{\sqrt2}\sum_{n\ge1}\frac{1}{n}\,f_n\otimes f_n$
and $c_j(\vartheta):=-\tfrac12\langle f_j,S^*(\vartheta)h_0\rangle_w$, i.e., $D^{1/2}\Upsilon^\vartheta=\tfrac{1}{\sqrt2}\sum_{j\ge1}\frac{1}{j}\,c_j(\vartheta)\,f_j$.
Moreover,
\begin{align}
  \langle D,\psi_2(t,u)\rangle
  &= -\frac{(\nu+i\lambda)^2}{8}\sum_{n=1}^{\infty}\frac{1}{n^4}\int_0^t f_n(s+\vartheta)^2\,\D s
     \;-\;\frac{(\nu+i\lambda)}{4}\sum_{n=1}^{\infty}\frac{c_n(\vartheta)}{n^4}\int_0^t f_n(s+\vartheta)\,\D s,
  \label{eq:D-psi2}
\end{align}
and 
\begin{align}
    \phi(t, u) &= \beta t - \beta \int_0^t \exp\left\{\frac{1}{8} (\nu+\mathrm{i}\lambda)^2\sum_{k=1}^\infty\frac{1}{k^4}\int_0^s f_k(l+\vartheta)\left(f_k(l+\vartheta)+\frac{2c_k(\vartheta)}{\nu+\mathrm{i}\lambda}\right)\, \D l\right\}\, \D s \label{eq:riccati_pure_jump3} \, . 
\end{align}
\end{lemma}
\begin{proof}
By the semigroup property, $\psi_1(t,u)=S^*(t)u_1$, which gives \eqref{eq:riccati_pure_jump1}. 
For the quadratic/linear drivers, write
$$
  D^{1/2}\psi_1(s,u)=\frac{\nu+i\lambda}{\sqrt2}\sum_{n\geq 1}\frac{1}{n}\,f_n(s+\vartheta)\,f_n,
  \qquad
  D^{1/2}\Upsilon^\vartheta=\frac{1}{\sqrt2}\sum_{j\geq 1}\frac{1}{j}\,c_j(\vartheta)\,f_j.
$$
Using $(a\otimes b)h=\langle a,h\rangle_w\,b$ and integrating 

$\partial_t\psi_2
= -\tfrac12(D^{1/2}\psi_1)^{\otimes2}-D^{1/2}\psi_1\otimes D^{1/2}\Upsilon^\vartheta$
yields \eqref{eq:riccati_pure_jump2}. 
For $\langle D,\psi_2\rangle=\sum_{n\geq 1 }\langle Df_n,\psi_2 f_n\rangle_w$, use
$Df_n=\tfrac{1}{2n^2}f_n$ and orthonormality to obtain \eqref{eq:D-psi2}. 
Finally, $\partial_t\phi=F(\psi_2)$ with 
$F(u)=-\beta\big(\E^{-\langle D,u\rangle}-1\big)$ gives \eqref{eq:riccati_pure_jump3}.
\end{proof}

From Theorem~\ref{prop:complex-extension} and Lemma~\ref{lem:Riccati-pure-jump-stoc-vol} (with $u_1=(\nu+i\lambda)u_\vartheta$, $\nu>1$) we obtain
\begin{equation}
\label{eq:option_b0}
\begin{aligned}
\mathbb{E}_{\MQ}\big[(F(t,t+\vartheta)-K)^+\mid \cF_0\big]
&= \frac{1}{2\pi}\int_{\R} g(\lambda)\,
   \exp\Big(
      -\phi(t,u) + \langle X_0,\psi_1(t,u)\rangle_w - \langle Y_0,\psi_2(t,u)\rangle
   \Big)\,\D\lambda \\
&= \frac{1}{2\pi}\int_{\R} g(\lambda)\,
   \exp\Big(
     -\beta t +  \beta \int_0^t \E^{-\langle D,\psi_2(s,u)\rangle}\,\D s
       \\
&\qquad\qquad
      +(\nu+i\lambda)\,X_0(t+\vartheta)- \langle Y_0,\psi_2(t,u)\rangle
   \Big)\,\D\lambda,
\end{aligned}
\end{equation}
where
$$
\langle D,\psi_2(s,u)\rangle
= -\frac{(\nu+i\lambda)^2}{8}\sum_{n=1}^{\infty}\frac{1}{n^4}\int_0^s f_n(r+\vartheta)^2\,\D r
  -\frac{(\nu+i\lambda)}{4}\sum_{n=1}^{\infty}\frac{c_n(\vartheta)}{n^4}\int_0^s f_n(r+\vartheta)\,\D r,
$$
and $c_n(\vartheta):=-\tfrac12\langle f_n,S^*(\vartheta)h_0\rangle_w$.

For numerical implementation we truncate the series at $N$:
\begin{equation}
\label{eq:option_b0_trunc}
\begin{aligned}
\mathbb{E}_{\MQ}\!\big[(F(t,t+\vartheta)-K)^+\mid \cF_0\big]
&\approx \frac{1}{2\pi}\int_{\R} g(\lambda)\,
   \exp\Big(
     -\beta t + \beta \int_0^t \E^{-\langle D,\psi_2^{(N)}(s,u)\rangle}\,\D s \\
&\qquad\qquad\qquad\qquad
      + (\nu+i\lambda)X_0(t+\vartheta)
      - \langle Y_0,\psi_2^{(N)}(t,u)\rangle
   \Big)\,\D\lambda,
\end{aligned}
\end{equation}
with 
$$
\langle D,\psi_2^{(N)}(s,u)\rangle
= -\frac{(\nu+i\lambda)^2}{8}\sum_{n=1}^{N}\frac{1}{n^4}\int_0^s f_n(r+\vartheta)^2\,\D r
  -\frac{(\nu+i\lambda)}{4}\sum_{n=1}^{N}\frac{c_n(\vartheta)}{n^4}\int_0^s f_n(r+\vartheta)\,\D r.
$$
The function $f_n$ can be efficiently evaluated by iteration, see Appendix~\ref{sec:appendix1}. The integrals $\int_0^s f_n(r+\vartheta)\,\D r$ and $\int_0^s f_n(r+\vartheta)^2\,\D r$ are computed numerically via Gauss--Legendre quadrature; in our implementation we use the \texttt{scipy.integrate.quad} routine, which provides adaptive Gaussian quadrature and is sufficient for the smooth integrands arising from the basis functions $f_n$.

  Tables \ref{tab:n_theta_b0} and \ref{tab:n_beta_b0} present the convergence analysis on the number of basis functions used, i.e., $N$. We use $N=10$ as the benchmark, and compute the relative absolute differences in the option price for $N=2, 3, 5, 8$ compared to $N=10$, with different time to maturity $\vartheta$ (in Table \ref{tab:n_theta_b0}) and different intensity $\beta$ (in Table \ref{tab:n_beta_b0}). The results show that the difference is already less than 0.02\% when $N=5$, which suggest the option price converges fast w.r.t the number of used basis function.  Moreover, Figure \ref{fig:mc_b0} shows the comparison on the option prices between the Monte Carlo simulation and the analytical formula. One easily sees that the forward curve generated by Monte Carlo simulation converges to the curve calculated by our analytical formula.

\begin{table}[ht]
    \centering
    \begin{tabular}{|c|ccccc|}
        \hline
        \textbf{Time to maturity} & \textbf{1} & \textbf{3} & \textbf{5} & \textbf{10} & \textbf{20} \\ \hline
        \textbf{N=2} & 0.06\% & 0.09\% & 0.33\% & 0.63\% & 0.65\% \\ 
        \textbf{N=3} & 0.01\% & 0.04\% & 0.03\% & 0.05\% & 0.04\% \\ 
        \textbf{N=5} & 0.01\% & 0.02\% & 0.01\% & 0.00\% & 0.00\% \\ 
        \textbf{N=8} & 0.00\% & 0.00\% & 0.00\% & 0.00\% & 0.00\% \\ \hline
    \end{tabular}
    \vspace{2mm}
    \caption{Relative absolute difference for different number of basis function $n$ and time to maturity $\vartheta$ on the option price compared to $n=10$, with $Y_0=D, h_0=1, \alpha=0.1, T_0=1$, strike $K=1$ and the intensity $\beta=1$. Pure jump case.}
    \label{tab:n_theta_b0}
\end{table}

\begin{table}[ht]
    \centering
    \begin{tabular}{|c|cccc|}
        \hline
        \textbf{beta} & \textbf{0.5} & \textbf{1} & \textbf{2} & \textbf{5} \\ \hline
        \textbf{N=2} & 1.46\% & 1.26\% & 1.01\% & 0.60\% \\
        \textbf{N=3} & 0.10\% & 0.05\% & 0.02\% & 0.02\% \\ 
        \textbf{N=5} & 0.01\% & 0.01\% & 0.01\% & 0.01\% \\ 
        \textbf{N=8} & 0.01\% & 0.01\% & 0.01\% & 0.00\% \\ \hline
    \end{tabular}
      \vspace{2mm}
    \caption{Relative absolute difference for different $n$ and $\beta$ on the option price compared to $n=10$, with $Y_0=D, h_0=1, \alpha=0.1, T_0=2$, strike $K=2$ and the time to maturity $\vartheta=1$. Pure jump case.}
    \label{tab:n_beta_b0}
\end{table}

\begin{figure}[ht]
    \centering
    \includegraphics[width=1\linewidth]{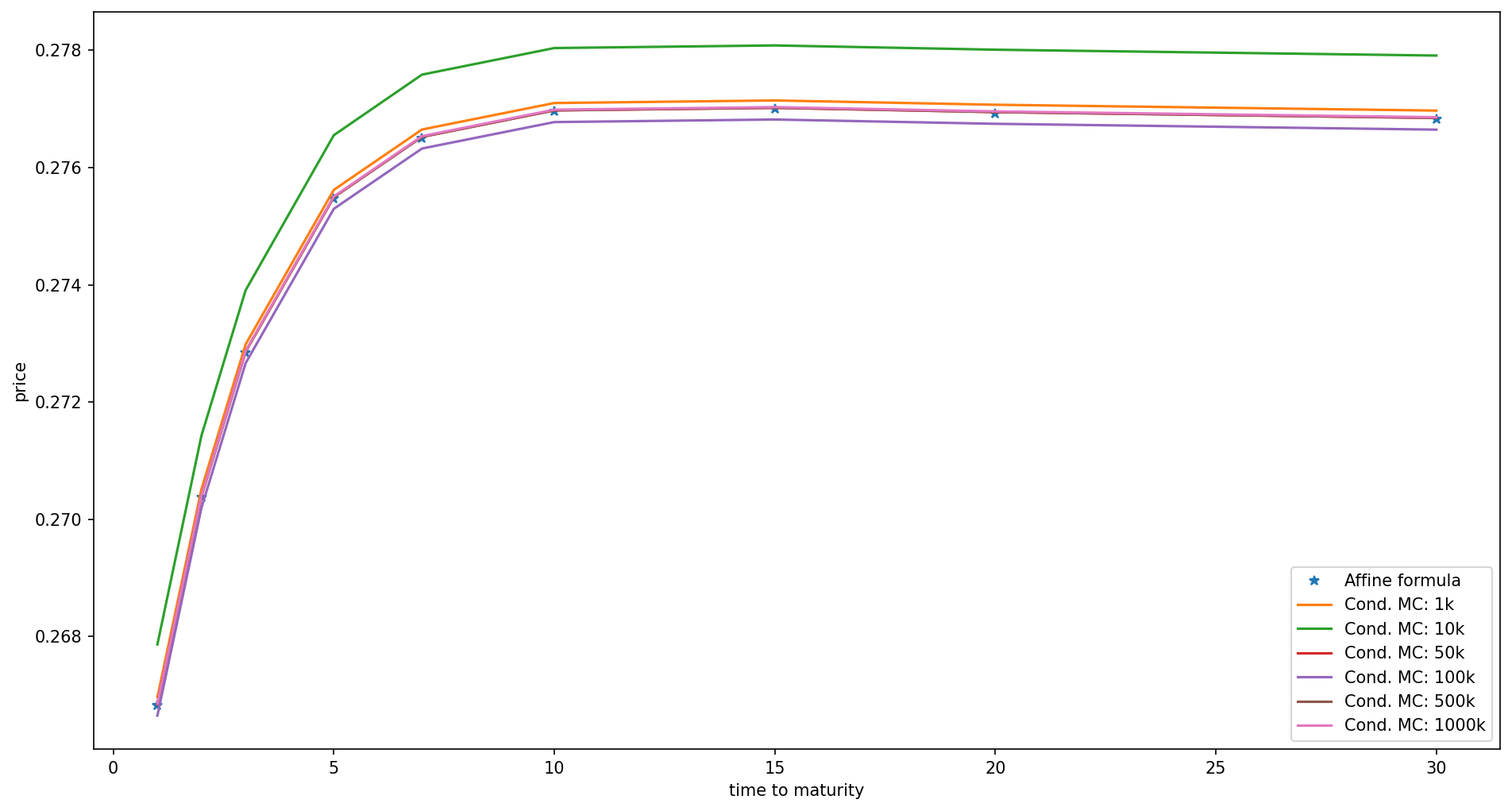}
    \caption{Conditional Monte Carlo vs.\ affine formula, with $Y_0=D, h_0=1, \alpha=0.1, T_0=1, N=5$, strike $K=1$ and the intensity $\beta=1$. }
    \label{fig:mc_b0}
\end{figure}

\subsection{Pricing in the operator-valued Barndorff–Nielsen–Shephard stochastic volatility model}\label{sec:example3}

Let $\cA$, $D$, $(\widetilde W_t)_{t\ge0}$, and $(J_t)_{t\ge0}$ be as in Example~\ref{sec:example2}, with
$J_t=(N_t-\beta t)D$ the compensated compound Poisson martingale of intensity $\beta>0$, independent of $\widetilde W$. 
Let $\cH_w$ be the space of self-adjoint Hilbert–Schmidt operators on $H_w$ and $\cH_w^+$ its positive cone. We consider the operator-valued BNS dynamics
\begin{align}\label{eq:OU-bns}
  \D X_t &= \big(\cA X_t + D^{1/2}\sY_t D^{1/2}\Upsilon^\vartheta\big)\,\D t
            + D^{1/2}\sY_t^{1/2}\,\D \widetilde W_t, 
            \qquad X_0=x\in H_w, \nonumber\\
  \D \sY_t &= \big(Q\sY_t + \sY_t Q^*\big)\,\D t + \D J_t, 
            \qquad \sY_0=y\in \cH_w^+,
\end{align}
where $Q\in\cL(H_w)$ is a (mean-reverting) drift operator (we will take it diagonal below) and 
$\Upsilon^\vartheta=-\tfrac12 S^*(\vartheta)h_0$ as in~\eqref{eq:upsilon}. The associated complex Riccati system \eqref{eq:complex-Riccati-phi}–\eqref{eq:complex-Riccati-psi-2} becomes
$$
  \partial_t \phi(t,u)=F(\psi_2(t,u)),\qquad 
  \psi_1(t,u)=u_1+\cA^*\int_0^t \psi_1(s,u)\,\D s,\qquad
  \partial_t \psi_2(t,u)=R(\psi_1(t,u),\psi_2(t,u)),
$$
with maps
\begin{equation*}
  R(h,u) = u Q^* + Q u - \tfrac12\,(D^{1/2}h)^{\otimes 2} - D^{1/2}h\otimes D^{1/2}\Upsilon^\vartheta,
  \qquad
  F(u) = -\beta\big(\E^{-\langle D,u\rangle}-1\big),
\end{equation*}
where $\langle\cdot,\cdot\rangle$ denotes the Hilbert–Schmidt pairing on $\cH_w$.

\begin{lemma}\label{lem:OU-bns-riccati-clean}
Let $w(x)=\E^{\alpha x}$ with $\alpha>0$ and $(f_n)_{n\in\N}$ be the orthonormal basis from \eqref{eq:f_basis}. 
Assume the OU drift is diagonal,
$$
  Q \;=\; -\sum_{n\ge1} a_n\, f_n\otimes f_n, 
  \qquad a_n\ge0,
$$
and set $D^{1/2}=\tfrac{1}{\sqrt2}\sum_{n\ge1}\tfrac{1}{n}\,f_n\otimes f_n$, 
$c_j(\vartheta):=-\tfrac12\langle f_j,S^*(\vartheta)h_0\rangle_w$ (so $D^{1/2}\Upsilon^\vartheta=\tfrac{1}{\sqrt2}\sum_{j\ge1}\tfrac{1}{j}c_j(\vartheta)\,f_j$). 
For $u_1=(\nu+i\lambda)u_\vartheta$ with $\nu>1$, $\lambda\in\R$, the complex Riccati solution to \eqref{eq:OU-bns} satisfies
\begin{align}
  \psi_1(t,u)(x) 
  &= (\nu+i\lambda)\!\left(1+\frac{1}{\alpha}\big(1-\E^{-\alpha(x\wedge t)}\big)
                               +\frac{1}{\alpha}\E^{-\alpha t}\,\mathbf 1_{\{x\geq t\}}
                                  \big(1-\E^{-\alpha(\vartheta\wedge x-t)}\big)\right), \label{eq:psi1-bns-clean}\\[4pt]
  \psi_2(t,u) f_n 
  &= -\sum_{j\geq 1}\int_0^t \E^{-(a_n+a_j)(t-s)}
     \Big[ \frac{(\nu+i\lambda)^2}{4\,n j}\,f_n(s+\vartheta)f_j(s+\vartheta)
           + \frac{(\nu+i\lambda)}{2\,n j}\,c_j(\vartheta)\,f_n(s+\vartheta)\Big]\D s  f_j, 
  \label{eq:psi2-bns-clean}
\end{align}
and
\begin{align}
    \phi(t, u) &= \beta t - \beta \int_0^t \exp\big(\frac{1}{8} (\nu+\mathrm{i}\lambda)^2\sum_{n=1}^\infty\frac{1}{n^4}\E^{\frac{s}{n^2}} \int_0^s \E^{-\frac{l}{n^2}} f_n(l+\vartheta)\left(f_n(l+\vartheta)+\frac{2c_n(\vartheta)}{\nu+\mathrm{i}\lambda}\right)\, \D l\big)\, \D s.\label{eq:phi}
\end{align}
In particular, we have
\begin{equation}\label{eq:Dxi2}
  \langle D,\psi_2(t,u)\rangle
  = -\int_0^t\Big[
       \frac{(\nu+i\lambda)^2}{8}\sum_{n\geq 1}\frac{1}{n^4} f_n(s+\vartheta)^2
       + \frac{(\nu+i\lambda)}{4}\sum_{n\geq 1}\frac{c_n(\vartheta)}{n^4} f_n(s+\vartheta)
     \Big] \E^{-2a_n(t-s)}\,\D s.
\end{equation}
\end{lemma}
\begin{proof}
Since $\psi_1$ solves $\psi_1(t,u)=u_1+\cA^*\int_0^t \psi_1(s,u)\,\D s$ with $u_1=(\nu+i\lambda)u_\vartheta$,
we have the mild form $\psi_1(t,u)=S^*(t)u_1$, which yields \eqref{eq:psi1-bns-clean} exactly as in Lemma~\ref{lem:Riccati-pure-jump-stoc-vol}.

The $\psi_2$-equation reads
$$
  \partial_t\psi_2(t)
  = \underbrace{\psi_2(t)Q^*+Q\psi_2(t)}_{\mathcal L_Q(\psi_2(t))}
        \;-\;\frac{1}{2}\,(D^{1/2}\psi_1(t))^{\otimes 2}
        \;-\; D^{1/2}\psi_1(t)\otimes D^{1/2}\Upsilon^\vartheta,
  \qquad \psi_2(0)=0.
$$
Hence, by variation of constants,
\begin{equation}\label{eq:mild-psi2}
  \psi_2(t)
  = \int_0^t \E^{(t-s)\mathcal L_Q}\left[\,
           -\frac{1}{2}(D^{1/2}\psi_1(s))^{\otimes 2}
           - D^{1/2}\psi_1(s)\otimes D^{1/2}\Upsilon^\vartheta
        \right]\D s.
\end{equation}
Assume $Q=-\sum_{n\ge1} a_n\, f_n\otimes f_n$ (with $a_n\geq 0$). Then, for basis tensors,
$$
  \E^{(t-s)\mathcal L_Q}(f_n\otimes f_j) = \E^{-(a_n+a_j)(t-s)}\,f_n\otimes f_j.
$$
Moreover,
$$
  D^{1/2}\psi_1(s)
  = \frac{\nu+i\lambda}{\sqrt 2}\sum_{n\geq 1}\frac{1}{n}\,f_n(s+\vartheta)\,f_n,
  \qquad
  D^{1/2}\Upsilon^\vartheta \;=\; \frac{1}{\sqrt 2}\sum_{j\geq 1}\frac{1}{j}\,c_j(\vartheta)\,f_j,
  \quad c_j(\vartheta):=-\tfrac12\langle f_j,S^*(\vartheta)h_0\rangle_w.
$$
Using $(a\otimes b)h=\langle a,h\rangle_w\,b$ and \eqref{eq:mild-psi2}, we obtain for each $n$:
\begin{align*}
  \psi_2(t)f_n
  &= -\!\sum_{j\ge1}\!\int_0^t \E^{-(a_n+a_j)(t-s)}
     \Big[\frac{(\nu+i\lambda)^2}{4\,n j}\,f_n(s+\vartheta)f_j(s+\vartheta)
          +\frac{(\nu+i\lambda)}{2\,n j}\,c_j(\vartheta)\,f_n(s+\vartheta)\Big]\D s\; f_j,
\end{align*}
which is precisely \eqref{eq:psi2-bns-clean}. Since $D=\frac12\sum_{n\ge1}\frac{1}{n^2}\,f_n\otimes f_n$,
$$
  \langle D,\psi_2(t)\rangle
  = \sum_{n\ge1}\frac{1}{2n^2}\,\langle f_n,\psi_2(t)f_n\rangle_w.
$$
By orthonormality, only the $j=n$ terms survive in $\langle f_n,\psi_2(t)f_n\rangle_w$, giving
$$
  \langle f_n,\psi_2(t)f_n\rangle_w
  = -\int_0^t \E^{-2a_n(t-s)}
      \left[\frac{(\nu+i\lambda)^2}{4\,n^2}\,f_n(s+\vartheta)^2
           +\frac{(\nu+i\lambda)}{2\,n^2}\,c_n(\vartheta)\,f_n(s+\vartheta)\right]\D s,
$$
from which we deduce formula \eqref{eq:Dxi2}.
Finally, since for the compound-Poisson input without finite-variation drift it holds that
$\partial_t\phi(t,u)=F(\psi_2(t,u))=-\beta\big(\E^{-\langle D,\psi_2(t,u)\rangle}-1\big)$,
integrating over $[0,t]$ yields \eqref{eq:phi}.
\end{proof}

Therefore, the option price is
\begin{equation}\label{eq:option1_b0}
\begin{aligned}
\mathbb{E}_{\MQ}\big[(F(t,t+\vartheta)-K)^+\mid \cF_0\big]
&= \frac{1}{2\pi}\int_{\R} g(\lambda)\,
   \exp\Big(
      -\phi(t,u) + \langle X_0,\psi_1(t,u)\rangle_w - \langle Y_0,\psi_2(t,u)\rangle
   \Big)\,\D\lambda \\
&= \frac{1}{2\pi}\int_{\R} g(\lambda)\,
   \exp\Big(
      \beta\!\int_0^t \!\big(\E^{-\langle D,\psi_2(s,u)\rangle}-1\big)\,\D s\big)
      \\ &\qquad\qquad \times\exp\Big( (\nu+i\lambda)\,X_0(t+\vartheta)
      - \langle Y_0,\psi_2(t,u)\rangle\Big)\,\D\lambda,
\end{aligned}
\end{equation}
where $u_1=(\nu+i\lambda)u_\vartheta$ with $\nu>1$, and $\langle D,\psi_2(t,u)\rangle$ is as in \eqref{eq:Dxi2}, 
with $Q=-\sum_{n\ge1} a_n\,f_n\otimes f_n$ ($a_n\ge0$) and $c_n(\vartheta):=-\tfrac12\langle f_n,S^*(\vartheta)h_0\rangle_w$.
The term $\langle X_0,\psi_1(t,u)\rangle_w=(\nu+i\lambda)X_0(t+\vartheta)$ comes from $\psi_1(t,u)=S^*(t)u_1$.

For the numerics, we truncate the series at $N$:
\begin{equation}\label{eq:option1_b0_trunc}
\begin{aligned}
\mathbb{E}_{\MQ}\big[(F(t,t+\vartheta)-K)^+\mid \cF_0\big]
&\approx \frac{1}{2\pi}\int_{\R} g(\lambda)\,
   \exp\Big(
      \beta\!\int_0^t \!\big(\E^{-\langle D,\psi_2^{(N)}(s,u)\rangle}-1\big)\,\D s\Big) \\
      &\qquad\qquad\times \exp\Big((\nu+i\lambda)\,X_0(t+\vartheta)
      - \langle Y_0,\psi_2^{(N)}(t,u)\rangle
   \Big)\,\D\lambda,
\end{aligned}
\end{equation}
with
$$
\langle D,\psi_2^{(N)}(t,u)\rangle
= -\int_0^t\Big[
   \frac{(\nu+i\lambda)^2}{8}\sum_{n=1}^{N}\frac{1}{n^4} f_n(s+\vartheta)^2
 + \frac{(\nu+i\lambda)}{4}\sum_{n=1}^{N}\frac{c_n(\vartheta)}{n^4} f_n(s+\vartheta)
 \Big] \E^{-2a_n(t-s)}\,\D s.
$$
In our experiments, we specialize to $a_n=\frac{1}{2n^2}$, then $\E^{-2a_n(t-s)}=\E^{-(t-s)/n^2}$. \medskip

The following Figure \ref{fig:mc_b1} and Tables \ref{tab:n_theta_b1} and \ref{tab:n_beta_b1} present the Monte Carlo convergence and the convergence analysis of number of used basis function $N$. Similar to the results in the pure jump case, one observes the option price converges fast w.r.t the number of used basis function, and Monte Carlo simulation converges to our analytical formula.
\begin{table}[h]
    \centering
    \begin{tabular}{|c|ccccc|}
        \hline
        \textbf{Time to maturity} & \textbf{1} & \textbf{3} & \textbf{5} & \textbf{10} & \textbf{20} \\
        \hline
        \textbf{N=2} & 0.03\% & 0.07\% & 0.23\% & 0.42\% & 0.42\% \\
        \textbf{N=3} & 0.01\% & 0.03\% & 0.02\% & 0.03\% & 0.01\% \\
        \textbf{N=5} & 0.01\% & 0.02\% & 0.00\% & 0.00\% & 0.00\% \\
        \textbf{N=8} & 0.00\% & 0.00\% & 0.00\% & 0.00\% & 0.00\% \\
        \hline
    \end{tabular}
    \vspace{2mm}
    \caption{Relative absolute difference for different number of basis function $N$ and time to maturity $\vartheta$ on the option price compared to $N=10$, with $Y_0=D, h_0=1, \alpha=0.1, T_0=1$, strike $K=1$ and the intensity $\beta=1$. BNS model}
    \label{tab:n_theta_b1}
\end{table}

\begin{table}[h]
    \centering
    \begin{tabular}{|c|cccc|}
    \hline
    \textbf{beta} & \textbf{0.5} & \textbf{1} & \textbf{2} & \textbf{5} \\
        \hline
        \textbf{N=2} & 0.37\% & 0.32\% & 0.24\% & 0.03\% \\
        \textbf{N=3} & 0.03\% & 0.05\% & 0.10\% & 0.21\% \\
        \textbf{N=5} & 0.00\% & 0.00\% & 0.00\% & 0.00\% \\
        \textbf{N=8} & 0.00\% & 0.00\% & 0.00\% & 0.00\% \\
        \hline
    \end{tabular}
    \vspace{2mm}
    \caption{Relative absolute difference for different $N$ and $\beta$ on the option price compared to $N=10$, with $Y_0=D, h_0=1, \alpha=0.1, T_0=2$, strike $K=2$ and the time to maturity $\vartheta=1$. BNS model.}
    \label{tab:n_beta_b1}
\end{table}

\begin{figure}[ht]
    \centering
    \includegraphics[width=1\linewidth]{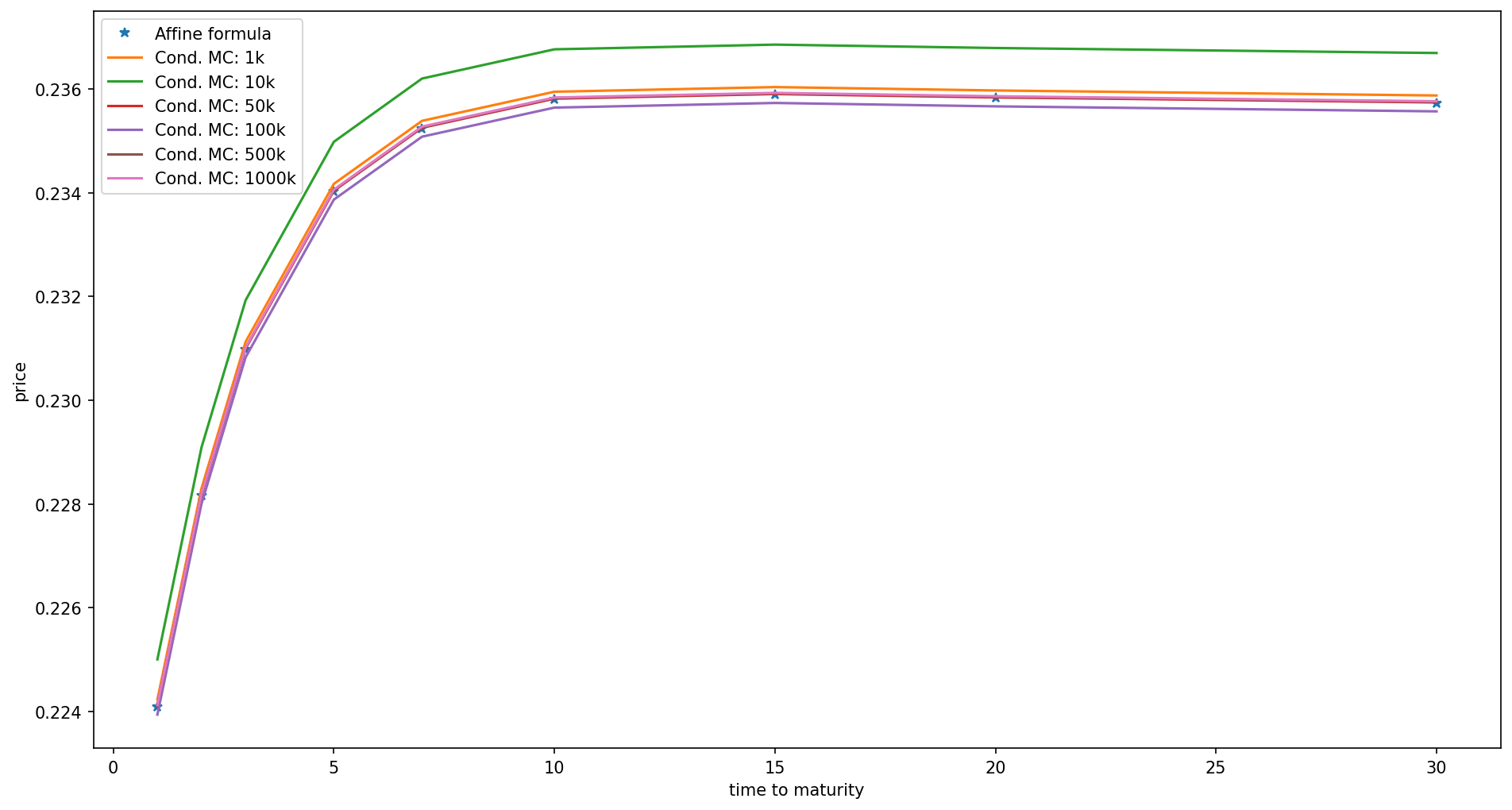}
    \caption{Conditional Monte Carlo vs.\ affine formula, with $Y_0=D, h_0=1, \alpha=0.1, T_0=1, N=5$, strike $K=1$ and the intensity $\beta=1$. }
    \label{fig:mc_b1}
\end{figure}

\subsection{Pricing in the affine jump model with state-dependent intensity}\label{sec:example4}

As a fourth example we validate the affine semi-explicit formula for the model from Section~\ref{sec:OU-jumps} in the case of a \emph{state-dependent} compound-Poisson jump kernel. The covariance process is now an operator-valued OU process driven by jumps whose arrival rate depends on the current state.

\paragraph{Model specification.}
We take $N=5$ basis functions $(f_n)_{n=1}^5$ as in \eqref{eq:f_basis} and set $D = \frac{1}{2}\sum_{n=1}^N n^{-2} f_n\otimes f_n$.
The covariance process $\mathsf{Y}_t = \sum_{n=1}^N y_n(t)\, f_n\otimes f_n$ satisfies
\begin{equation}\label{eq:state-dep-Y}
  \D\mathsf{Y}_t = -(Q\mathsf{Y}_t + \mathsf{Y}_t Q)\,\D t + \D J_t,
  \qquad \mathsf{Y}_0 = \sum_{n=1}^N \tfrac{1}{2n^2}\, f_n\otimes f_n,
\end{equation}
with $Q = \sum_{n=1}^N \kappa_n\, f_n\otimes f_n$, $\kappa_n = n^{-2}$, so each diagonal entry satisfies $\D y_n = -2\kappa_n y_n\,\D t + \D J_n$. The jump process $J_t$ has a single jump direction $\Xi = \sum_{n=1}^N \xi_n\, f_n\otimes f_n$, $\xi_n = (2n^2)^{-1}$, and state-dependent intensity
\begin{equation}\label{eq:state-dep-intensity}
  \lambda(\mathsf{Y}) = \lambda_0 + \sum_{n=1}^N \frac{c_n}{\|\xi\|_2^2}\, y_n,
  \qquad \lambda_0 = 0.05,\quad c_n = \frac{0.05}{n^2},\quad \|\xi\|_2^2 = \sum_{n=1}^N \xi_n^2.
\end{equation}
The forward dynamics~\eqref{eq:sto-vol-pure-jump} and HJM drift $\Upsilon^\vartheta = -\frac{1}{2}S^*(\vartheta)h_0$ are as in the previous examples.

\paragraph{Numerical Riccati scheme.}
Unlike the pure-jump L\'evy case (Lemma~\ref{lem:Riccati-pure-jump-stoc-vol}), no closed-form Riccati solution is available due to the state-dependent jump intensity. In the single-direction diagonal specification implemented below, the generalized Riccati equations from Theorem~\ref{prop:complex-extension} reduce, for the matrix coefficients $q_{ij}(\tau) = \langle f_i, \psi_2(\tau,u)f_j\rangle_w$, to the system
\begin{equation}\label{eq:state-dep-riccati}
  \partial_\tau q_{ij}
  = -(\kappa_i+\kappa_j)\,q_{ij}
  - \frac{c^2}{4\,ij}\,f_i(\vartheta+\tau)\,f_j(\vartheta+\tau)
  + \frac{c}{4\,ij}\,f_i(\vartheta+\tau)\,f_j(\vartheta)
  - \delta_{ij}\,\frac{c_i}{\|\xi\|_2^2}\!\left(\exp\Bigl(-\sum_{n=1}^N \xi_n\, q_{nn}\Bigr)-1\right),
\end{equation}
where $c = \nu+i\lambda$. We solve~\eqref{eq:state-dep-riccati} numerically by a forward Euler scheme with step $\delta=1/365$, and recover option prices via the Fourier inversion formula~\eqref{eq:chkk-call-option-price} with truncation $N_{\rm trunc}=10$.

\paragraph{Monte Carlo benchmark.}
Given a simulated covariance path $(y_n(t_l))$, obtained by Euler--Maruyama with state-dependent Poisson arrivals and step $\delta=1/365$, the log-forward $\log F(T,T+\vartheta)$ is conditionally Gaussian with conditional mean and variance
\begin{align*}
  \mu_{\rm cond} &= X_0(T+\vartheta)
    - \frac{\delta}{4}\sum_{l}\sum_{n=1}^N \frac{y_n(t_l)}{n^2}\,f_n(\vartheta)\,f_n(\vartheta+s_l),\\
  \sigma_{\rm cond}^2 &= \frac{\delta}{2}\sum_{l}\sum_{n=1}^N \frac{y_n(t_l)}{n^2}\,f_n(\vartheta+s_l)^2,
\end{align*}
where $s_l=(l+\frac12)\delta$. Option prices are obtained by averaging the Black formula over simulated paths. We use $\alpha=0.1$, $h_0=1$, $T=1$, $K=1$, and $N_{\rm sim}\in\{1\mathrm{k},10\mathrm{k},50\mathrm{k},100\mathrm{k},200\mathrm{k}\}$.

Figure~\ref{fig:mc_state_dep} shows the option prices as a function of time to maturity $\vartheta$. The Monte Carlo estimates converge to the semi-explicit affine prices as $N_{\rm sim}$ increases, confirming consistency of the numerical Riccati scheme with direct simulation of the state-dependent model.

\begin{figure}[ht]
  \centering
  \includegraphics[width=1\linewidth]{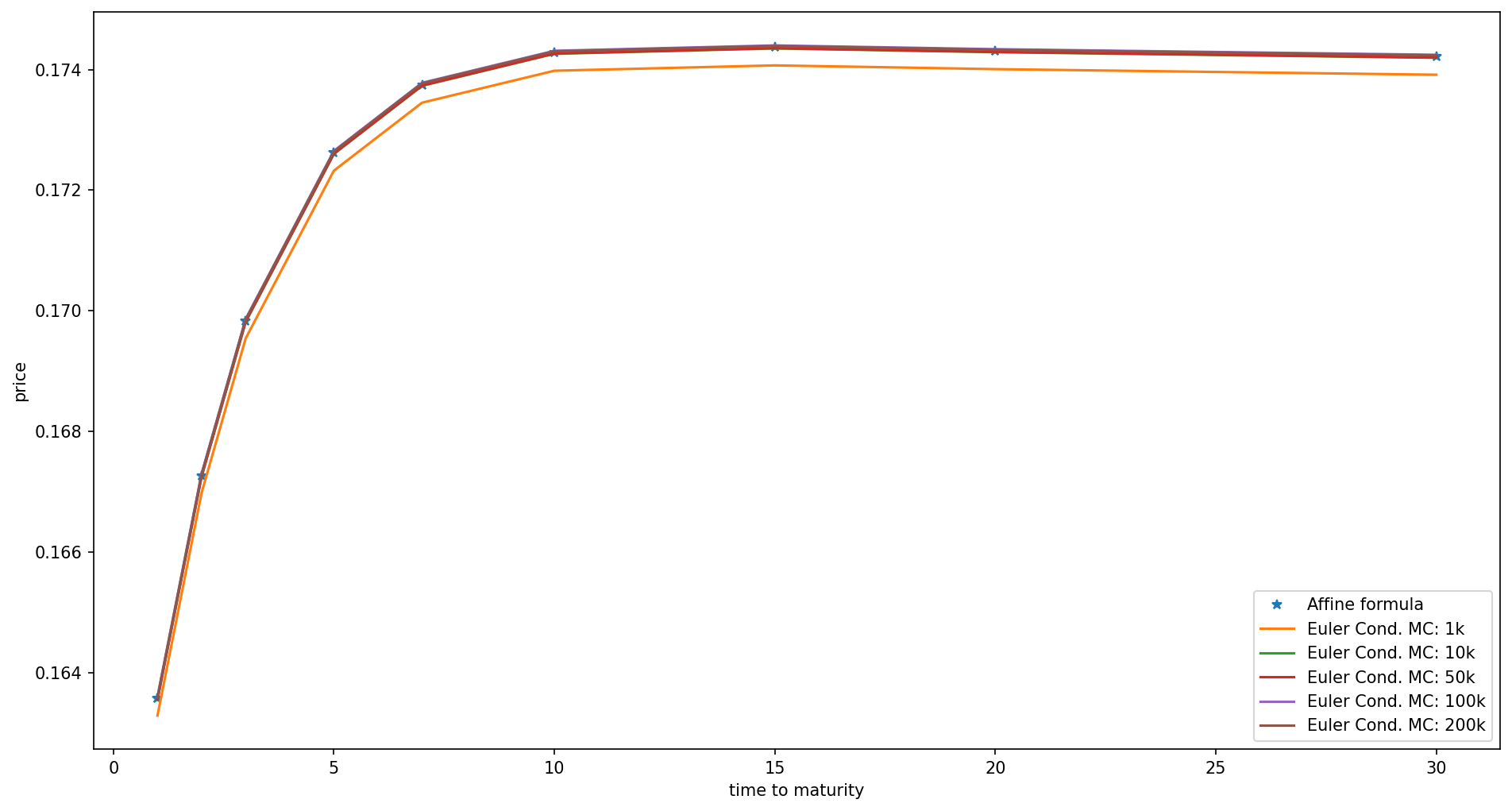}
  \caption{Euler conditional Monte Carlo vs.\ affine semi-explicit formula for the state-dependent jump model, with $\alpha=0.1$, $h_0=1$, $T=1$, $K=1$, $N=5$, $\lambda_0=0.05$, $\kappa_n=n^{-2}$, $\xi_n=(2n^2)^{-1}$.}
  \label{fig:mc_state_dep}
\end{figure}

\subsection{Running time analysis}
We report runtimes from the canonical root benchmark entry point
\texttt{python main\_new.py runtime\_benchmarks\_paper}, which stores the corresponding machine metadata and timings in \texttt{runtime\_benchmarks\_paper.json}. The benchmark was run on a Lenovo ThinkPad X1 Carbon Gen 11 (model 21HM006WMH) with a 13th Gen Intel Core i7-1355U processor and 32\,GB RAM, under Python~3.12.2, NumPy~2.2.5, and SciPy~1.15.2. For the explicit L\'evy and BNS models we use $100\,000$ conditional Monte Carlo paths, while for the state-dependent model we use $200\,000$ paths. Each entry in Table~\ref{tab:run_time} is a single wall-clock evaluation of the corresponding pricing routine in the canonical root implementation.

\begin{table}[h]
    \centering
    \begin{tabular}{|c|cc|}
    \hline
     Runtime (seconds)    & Affine formula & Monte Carlo\\
         \hline
     L\'evy driven model &  0.16 & 88.89\\
     BNS model   & 0.08 & 103.51\\
     State-dependent jump model & 55.01 & 71.78\\
     \hline
    \end{tabular}
        \vspace{2mm}
    \caption{Canonical runtime benchmark generated by \texttt{python main\_new.py runtime\_benchmarks\_paper} on the Lenovo ThinkPad X1 Carbon Gen 11 described above. The L\'evy and BNS rows use $100\,000$ Monte Carlo scenarios; the state-dependent row uses $200\,000$ scenarios.}
    \label{tab:run_time}
\end{table}

These timings confirm the qualitative picture already seen in the pricing figures: for the explicit L\'evy and BNS specifications, the affine formulas are orders of magnitude faster than the corresponding Monte Carlo benchmark, whereas in the state-dependent example the runtime gap is much smaller, due to the complex generalized Riccati equations which do not have closed form solutions beyond the OU case. Developing fast schemes for this type of infinite-dimensional Riccati equations is a promising research question left for the future.

\section{Computation of complex moments}\label{sec:complex-moments}

In this section, we prove the existence of solutions to the complex-extensions of the generalized Riccati equations~\eqref{eq:complex-Riccati-phi}-\eqref{eq:complex-Riccati-psi-2} and~\eqref{eq:complex-Riccati-phi-Wishart}-\eqref{eq:complex-Riccati-psi-2-Wishart}. These existence results are then used to compute exponential moments for the affine pure-jump stochastic volatility model and to derive an approximation for the Laplace transform of the Wishart stochastic volatility model. To this end, we introduce the concept of quasi-monotonicity.

\begin{definition}
    Let $V$ be a real Banach space, $K\subset V$ a wedge. Given $ D\subset V$ and $f\colon D\to V$ is said to be quasi monotone (with respect to $K$) if $\forall x,y \in D$ satisfying 
 $x\leq_{K}y$ and $\forall x' \in K^*$ satisfying  $\langle x', x-y\rangle_{V', V} =0$, it holds $\langle x', f(x)-f(y)\rangle_{V', V} \leq 0\,.$
\end{definition}

\subsection{Real and complex moments for the affine pure-jump stochastic volatility model}

In the following proposition we show the existence of a unique solution to the
extended Riccati equations~\eqref{eq:extended-Riccati-P}-\eqref{eq:extended-Riccati-q2} 
up to its maximal lifetime. 

\begin{proposition}\label{prop:extended-Riccati-existence}
Let $(b,B,m,\mu)$ be an admissible parameter set satisfying Assumptions \ref{def:admissibility} and let $(\cA, \dom(\cA))$ be the
generator of a strongly continuous semigroup, let $D \in \cL(H)$ be a positive
self-adjoint operator, and let $\Upsilon \in H$. Then for every $u=(u_{1},u_{2})\in H\times \cO$ there
exist a positive real number $T_{\sf q_{2}}$ and a unique solution $({\sf p}(\cdot,u), {\sf q}_1(\cdot,u),
{\sf q}_2(\cdot,u))$ to \eqref{eq:extended-Riccati-P}-\eqref{eq:extended-Riccati-q2}
on $[0,T_{\sf q_{2}})$.
\end{proposition}
\begin{proof}
Standard semigroup theory ensures that for any $u_{1}\in H$ and $T>0$ the unique mild solution to
\eqref{eq:extended-Riccati-q1} is given by ${\sf q}_{1}(t,u)=S^*(t)u_{1}$ for $t\in
[0,T]$. We define $\cR_{u_1} (t, \cdot)\colon \cO \to \cH$, by 
\begin{align*}
\cR_{u_1}(t,u) &= 
\tilde{B}^{*}(u) 
-
\tfrac{1}{2}\big(D^{1/2}S^*(t)u_1\big)^{\otimes 2}- D^{1/2}S^*(t)u_1\otimes D^{1/2}\Upsilon \\
&\qquad - \int_{\cHpluso}\big(\E^{-\langle
    \xi,u\rangle}-1\big)\frac{\mu(\D \xi)}{\norm{\xi}_2^{2}}.
\end{align*}
Plugging ${\sf q}_{1}(t,u)$ into \eqref{eq:extended-Riccati-q2}, we thus obtain the equation
\begin{align}\label{eq:extended-Riccati-q-2-1-proof}
  \begin{cases}
    \frac{\partial {\sf q}_{2}}{\partial t}(t,u)= \cR_{u_1}(t,{\sf q}_2(t,u))\,,\\
    {\sf q}_{2}(0,u)=u_{2}\,. 
  \end{cases}
\end{align}
Observe that the function
$\cR_{u_{1}}(t,\cdot)$ is locally Lipschitz continuous on $\cO$ for every
$t\geq 0$ and $u_{1}\in H$. Since $\cO$ is assumed to be open, it follows from
standard ODE results (see
e.g.~\cite[Chapter 6, Proposition 1.2]{Mar76}) that for every $u_{2}\in\cO$ there exist a $T_{\sf q_{2}}>0$ and a unique solution
${\sf q}_{2}(\cdot,u)$ of~\eqref{eq:extended-Riccati-q-2-1-proof} on $[0,T_{{\sf q}_{2}})$ with
\begin{align*}
T_{\sf q_{2}}=\liminf_{n\to\infty}\set{t\geq 0\colon \norm{{\sf q}_{2}(t,u)}_2\geq
  n\text{ or }{\sf q}_{2}(t,u)\in\partial \cO}.
\end{align*}
By inserting $\sf q_2(\cdot,u)$ in  \eqref{eq:extended-Riccati-P} and observing that
$F$ is continuous on $\cO$, the statement follows.
\end{proof}

To compute exponential moments for the process $(X_t, \sY_t)_{t\geq 0}$ in \eqref{eq:joint-affine-SDE}, we first consider
the joint process $(X^{(n)}_t, \sY_t)_{t\geq 0}$ obtained by replacing $\mathcal{A}$ in~\eqref{eq:joint-affine-SDE} by its Yosida approximation $\mathcal{A}^{(n)}:=n\mathcal{A}(nI-\mathcal{A})^{-1}$. Namely,
we consider the process $X^{(n)}\colon [0,\infty)\times \Omega \rightarrow H$ given by the solution to (see \cite[Proposition 6.4]{DZ92} for the existence of the solution to $X^{(n)}$)
\begin{align}\label{eq:approximating-Y}
X^{(n)}_{t}=x+\int_{0}^{t} (\cA^{(n)}X^{(n)}_{s}+ D^{1/2}\sY_{s}D^{1/2}\Upsilon) \,\D s+ \int_{0}^{t} D^{1/2}\sY_{s}^{1/2}\,\D W_{s}\,, \quad t\geq 0.
\end{align}
The use of the approximation allows us to exploit the semimartingale theory and to apply the It\^o formula and standard techniques in order to compute exponential moments for 
 $(X^{(n)}_t, \sY_t)_{t\geq 0}$. Then we take the limit when $n$ goes to $\infty$. Observe that from \cite[Proposition 7.5]{DZ92} it holds
 \begin{equation}\label{eq:Yosida_OU_converges}
\lim_{n \rightarrow \infty} 
\mathbb{E} \left[ 
\sup_{0 \leq t \leq T} \| X^{(n)}_t -X_t \|^2_{H} 
\right]= 0\,.
\end{equation}
We are ready to derive the exponential moments of the affine pure-jump stochastic volatility model. 
\begin{proposition}\label{prop:extended-affine-Z}
Let $(b,B,m,\mu)$ be an admissible parameter set satisfying Assumptions~\ref{def:admissibility} and~\ref{ass:chkk-finite-variation},
let $(X,Y)$ be the stochastic volatility model \eqref{eq:joint-affine-SDE},
let $u=(u_1, u_2) \in H\times \cO$, and let $({\sf P}(\cdot,u), {\sf q}_1(\cdot, u), {\sf q}_2(\cdot,u))$ be the mild solution to the Riccati equation \eqref{eq:extended-Riccati} up to time 
$T$ the existence of which is guaranteed by Proposition \ref{prop:extended-Riccati-existence}.
Then we have:
  \begin{align}
    \label{eq:extended-expo-moment-finite}
    \EX{\E^{\langle X_{T}, u_{1}\rangle_{H}-\langle {\sf Y}_{T}, u_{2}\rangle}}<\infty.
  \end{align}
  Moreover, for every $t\leq T$, it holds
  \begin{align}
    \label{eq:extended-affine-formula-intro}
    \EX{\E^{\langle X_{T}, u_{1}\rangle_{H}-\langle {\sf Y}_{T},
    u_{2}\rangle}\mid \mathcal{F}_t}=\E^{-{\sf P}(T-t,u)+\langle X_t, {\sf q}_{1}(T-t,u)\rangle_{H}-\langle {\sf Y}_t, {\sf q}_{2}(T-t,u)\rangle}. 
  \end{align}  
\end{proposition}
\begin{proof}
Let $({\sf P}^{(n)}(\cdot, u), {\sf q}_1^{(n)}(\cdot, u), {\sf q}_2^{(n)}(\cdot,u))$ be the solution to \eqref{eq:extended-Riccati} with $\cA =\cA^{(n)}$ (the $n$th Yosida approximation). We define the function $g_u^{(n)}(t,y,x) \colon[0,T]\times H\times \cHplus \rightarrow \R$ as follows 
$$g_u^{(n)}(t,x,y) = \exp\left\{-{\sf P}^{(n)}(T-t, u)+ \langle x, {\sf q}^{(n)}_1(T-t,u)\rangle_H-\langle y,{\sf q}^{(n)}_2(T-t,u)\rangle\right\}\,.$$
Similarly to the proof of \cite[Theorem 3.3]{CKK22b}, we apply the It\^o's formula to $g_{u}^{(n)}(t,X^{(n)}_t, {\sf Y}_t)$ to deduce that it is a local martingale. 

  Then observing that $g_{u}^{(n)}(t,X_{t}^{(n)}, {\sf Y}_t)$ is strictly positive for all $t \in [0,T]$, we infer it is a
  supermartingale and we have
    \begin{align*}
  \EX{\E^{\langle X^{(n)}_{T}, u\rangle_H-\langle {\sf Y}_{T}, u\rangle }}&= \EX{g_{u}^{(n)}(T,X_{T}^{(n)}, {\sf Y}_T)}\\
    &\leq g_u^{(n)}(0,x,y) < \infty\,, \quad (x,y) \in H\times \cH\,,
  \end{align*}
  from which it follows that $(g_{u}^{(n)}(t,X_{t}^{(n)}, {\sf Y}_t))_{t\in [0,T]}$ is a martingale and hence it holds
  \begin{align*}
  \EX{\E^{\langle X_{T}^{(n)},u_1\rangle_H - \langle {\sf Y}_T, u_2\rangle} \mid \mathcal{F}_t}
    &=\E^{- {\sf P}^{(n)}(T-t,u)+\langle X_t, {\sf q}_1^{(n)}(T-t,u)\rangle_H - \langle {\sf Y}_t, {\sf q}_2^{(n)}(T-t,u)\rangle}.  
  \end{align*}

Following similar arguments as in the proof of \cite[Proposition 3.5]{CKK22b}, we deduce that
$$\lim_{n\rightarrow \infty} \sup_{t \in [0,T]}|{\sf P}^{(n)}(t,u) - {\sf P}(t,u)| =0$$
and 
$$\lim_{n\rightarrow \infty}\left( \sup_{t \in [0,T]} \|{\sf q}_1^{(n)}(t,u) - {\sf q}_1(t,u)\|_H + \sup_{t \in [0,T]} \|{\sf q}_2^{(n)}(t,u) - {\sf q}_2(t,u)\|_2\right) = 0\,. $$
Hence taking limits for $n \rightarrow \infty$ and invoking \eqref{eq:Yosida_OU_converges} yields the statement.
\end{proof}

Before proving the existence and uniqueness of a solution to \eqref{eq:complex-Riccati-phi}-\eqref{eq:complex-Riccati-psi-2}, we give the following
lemma: 
\begin{lemma}\label{lem:growth-inequality}
Let ${\sf R}$ be as in \eqref{eq:R}. Then there exists a locally Lipschitz continuous function $g$ on $\cH$ such that
for all $h\in H^{\MC}$ and $u\in S(\cO)$ we have
\begin{align}\label{eq:growth-inequality}
  \Re(\langle \overline{u} , {\sf R}(h,u)\rangle)\leq g(\Re(u))(1+\norm{u}_{\cH^{\mathbb{C}}}^{2})(1+\norm{h}_{H^{\mathbb{C}}}),
\end{align}
\end{lemma}
\begin{proof}
  We first split $\Re(\langle \overline{u} , {\sf R}(u,h)\rangle)$ into three parts as follows:
  \begin{align}\label{eq:growth-inequality-1}
   \Re(\langle \overline{u},R(h,u)\rangle)&= I_1 + I_2+ I_3\,,
   \end{align}
   where
   \begin{align*}
   I_1&=  \Re(\langle \overline{u},
    \tilde{B}^{*}(u)\rangle)+\Re(\langle \overline{u}, D^{1/2}h\otimes D^{1/2}h\rangle) + \Re(\langle \overline{u}, D^{1/2}h\otimes D^{1/2}\Upsilon \rangle)\,, \\
    I_2 &= -\Re\Big(\int_{\cHplus\cap\set{0<\norm{\xi}_2\leq 1}}\big(\E^{-\langle
      u,\xi\rangle}-1\big)\langle\frac{\mu(\D\xi)}{\norm{\xi}_2^{2}},\overline{u}\rangle\Big)\,,\\
    I_3 &=  -\Re\Big(\int_{\cHplus\cap\set{\norm{\xi}_2>1}}\big(\E^{-\langle
      u,\xi\rangle}-1\big)\langle\frac{\mu(\D\xi)}{\norm{\xi}_2^{2}},\overline{u}\rangle\Big).
  \end{align*}
  For the first term on the right-hand side of~\eqref{eq:growth-inequality-1}, it holds
  \begin{align*}
    I_1&\leq
                      \norm{B^{*}}_{\cL(\cH)}\norm{u}^{2}_{\cH^{\mathbb{C}}}+\norm{D}_{\cL(H)}\norm{h}_{H^{\MC}}^{2}\norm{u}_{\cH^{\mathbb{C}}} + \norm{D}_{\cL(H)}\norm{h}_{H^{\MC}}\norm{\Upsilon}_{H}\norm{u}_{\cH^{\mathbb{C}}}\\
                    &\leq g_{1}(1+\norm{u}_{\cH^{\mathbb{C}}}^{2})(1+\norm{h}_{H^{\MC}}^{2}), 
  \end{align*}
  where $g_{1}=\norm{B^{*}}_{\cL(\cH)}+\norm{D}_{\cL(H)} +\norm{D}_{\cL(H)} \norm{\Upsilon}_H$. 
  Next observe that 
  \begin{align}\label{eq:big-jumps}
  \int_{\cHplus\cap\set{\norm{\xi}_2>1}}\E^{-\langle
    u,\xi\rangle} \frac{\mu(\D\xi)}{\norm{\xi}_2^{2}} &\leq_{\cHplus} \int_{\cHplus\cap\set{\norm{\xi}_2>1}}\E^{-\langle
    \Re(u),\xi\rangle} \frac{\mu(\D\xi)}{\norm{\xi}_2^{2}}
  \end{align}
  and 
  \begin{align}\label{eq:small-jumps}
    \int_{\cHplus\cap\set{0<\norm{\xi}_2\leq 1}} (\E^{-\langle
    u,\xi\rangle} -1)\frac{\mu(\D\xi)}{\norm{\xi}_2^{2}}
                                                                               &=
                                                    \int_{\cHplus\cap\set{0<\norm{\xi}_2\leq
                                                    1}}\int_{0}^{1}\langle
                                                    \xi,u\rangle \E^{-\langle
                                                    \xi,u\rangle}\D s\frac{\mu(\D\xi)}{\norm{\xi}_2^{2}}\nonumber\\ 
    &\leq_{\cHplus} \E^{\norm{\Re(u)}_2}\norm{u}_{\cH^\mathbb{C}}\int_{\cHplus\cap\set{0<\norm{\xi}_2\leq 1}}\norm{\xi}_2^{-1}\mu(\D\xi).
  \end{align}
  Inequality \eqref{eq:big-jumps} and the monotonicity of $\cHplus$ imply 
  \begin{align*}
    I_2&\leq
                                                                                \norm{\int_{\cHplus\cap\set{\norm{\xi}_2>1}}\E^{-\langle
                                                                                \Re(u),\xi\rangle}\frac{\mu(\D\xi)}{\norm{\xi}_2^{2}}}\norm{u}_{\cH^{\mathbb{C}}}
    +\norm{\mu(\cHplus\cap\set{\norm{\xi}_2>1})}_2\norm{u}_{\cH^{\mathbb{C}}}\\
                                                                              &\leq g_{2}(\Re(u))(1+\norm{u}_{\cH^{\mathbb{C}}}^{2}),
  \end{align*}
  where $g_{2}(u)=\norm{\int_{\cHplus\cap\set{\norm{\xi}_2>1}}\E^{-\langle
    \xi, u\rangle}\frac{\mu(\D\xi)}{\norm{\xi}_2^{2}}}+\norm{\mu(\cHplus\cap\set{\norm{\xi}_2>1})}$.
  Finally, for the third term in the right hand side of~\eqref{eq:growth-inequality-1}, it follows from inequality \eqref{eq:small-jumps} and the monotonicity of $\cHplus$
   \begin{align*}
  I_3 &\leq
   \norm{\int_{\cHplus\cap\set{0<\norm{\xi}_2\leq 1}}(\E^{-\langle u,\xi\rangle}-1)\frac{\mu(\D\xi)}{\norm{\xi}_2^{2}}}_{2}\norm{u}_{\cH^{\mathbb{C}}}.\\  
   &\leq\E^{\norm{\Re(u)}_2}\norm{\int_{\cHplus\cap\set{0<\norm{\xi}_2\leq
    1}}\norm{\xi}_2^{-1}\mu(\D\xi)}\norm{u}_{\cH^{\mathbb{C}}}^2.
  \end{align*}
  Hence setting
  $g_{3}(u)\df\E^{\norm{u}_{\cH^{\mathcal{C}}}}\norm{\int_{\cHplus\cap\set{0<\norm{\xi}_2\leq
        1}}\norm{\xi}_2^{-1}\mu(\D\xi)}_2$, we find that~\eqref{eq:growth-inequality} holds for the
  function $g=g_{1}+g_{2}+g_{3}$, which is locally Lipschitz continuous. 
\end{proof}

In the following proposition we prove the existence of a unique solution to \eqref{eq:complex-Riccati-phi}-\eqref{eq:complex-Riccati-psi-2} on $[0,T]$:
\begin{proposition}\label{prop:existence-complex-solution}
Let $(b,B,m,\mu)$ be an admissible parameter set satisfying Assumptions \ref{def:admissibility} and \ref{ass:chkk-finite-variation},
let $(\cA,\dom(\cA))$ be the generator of a strongly continuous semigroup, let $D \in
\cL(H)$ be a positive self-adjoint operator, and let $\Upsilon \in H$. Let $(u_1, u_2) \in H^\mathbb{C} \times S(\cO)$. Assume that the
extended generalized Riccati
equations~\eqref{eq:extended-Riccati-P}-\eqref{eq:extended-Riccati-q2} have
a unique solution on $[0,T]$, for the initial condition \sloppy $u= (\cR (u_1),\cR (u_2))$. Then there exists a unique
solution $(\phi(\cdot,u), \psi_1(\cdot,u), \psi_2(\cdot,u))$ to
\eqref{eq:complex-Riccati-phi}-\eqref{eq:complex-Riccati-psi-2} on $[0,T]$, with initial value $u=(u_1,0)$.     
\end{proposition}
\begin{proof}
  Let $u=(u_{1},u_{2})\in H^{\MC}\times S(\cO)$ and $T\geq
  0$. By standard semigroup theory, we see that the unique mild solution
  to~\eqref{eq:complex-Riccati-psi-2} on $[0,T]$ is given by
  $\psi_1(t,u)=S^{*}(t)u_{1}\in H^{\MC}$. For every $t\in [0,T]$ we consider 
the extension of the map $\cR_{u_{1}}$, from the proof of
  Proposition~\ref{prop:extended-Riccati-existence}, to a function $\cR_{u_1}
  (t, \cdot)\colon S(\cO)\to \cH^\MC$. Then by plugging
  $\psi_{1}(t,u)$ into \eqref{eq:extended-Riccati-q2}, we obtain the equation 
  \begin{align}\label{eq:extended-Riccati-q-2-1}
    \begin{cases}
    \frac{\partial \psi_{2}}{\partial t}(t,u)&= \cR_{u_1}(t,\psi_2(t,u))\,,\\
    \psi_{2}(0,u)&=0.
    \end{cases}
  \end{align}
  The map $\mathcal{R}_{u_{1}}(t,\cdot)$ is locally Lipschitz continuous on
  $S(\cO)$, which again by standard arguments implies the existence
  of a unique solution $\psi_{2}(\cdot,u)$ on some interval
  $[0,T_{\psi_{2}})$ with values in $S(\mathcal{O})\subseteq\cH^{\MC}$. We
  want to show that the lifetime $T_{\psi_{2}}$ of $\psi_{2}(\cdot, u_1, 0)$
  is always greater or equal than the lifetime 
  $T$ of $q_{2}(\cdot,\Re(u_1),0)$. 
  First observe that
  $-\big( D^{1/2}S^{*}(t)\Im(u))^{\otimes 2}\leq_{\cHplus} 0$ and 
  \begin{align*}
    \int_{\cHpluso}\Re (\E^{-\langle \xi , u\rangle } -1) \frac{\mu(\D\xi)}{\norm{\xi}_2^{2}}\leq_{\cHplus} \int_{\cHpluso}(\E^{-\langle \xi , \Re{(u)}\rangle } -1)\frac{\mu(\D\xi)}{\norm{\xi}_2^{2}}.
  \end{align*}
  This together with the fact that $q_{2}(0,\Re(u_1),0)=\Re(\psi_{2}(0,u_1,0))=0$ imply for all $t<T\wedge T_{\psi_{2}}$
  \begin{align}\label{eq:comparison}
   & \frac{\partial
    q_{2}(t,\Re(u_1),0)}{\partial t}-\cR_{\Re(u_{1})}\big(t,q_{2}(t,\Re(u_1),0)\big)\nonumber\\
    &\qquad =\frac{\partial \Re\psi_{2}(t,u_1,0)}{\partial t}-\Re \cR_{u_{1}}\big(t,\psi_{2}(t,u_1,0)\big)\nonumber\\ 
    &\qquad  \leq_{\cHplus} \frac{\partial \Re\psi_{2}(t,u_1,0)}{\partial t}-\cR_{\Re(u_{1})}\big(t,\Re\psi_{2}(t,u_1,0)\big)\,.
  \end{align}
 Note that $\cR_u(t, \cdot)$ is quasi-monotone with respect to $\cH^+$, for $u\in H$ (see \cite[Section 6]{Vol73} for the notion of quasi-monotonicity). Then by Volkmann's comparison result \cite[Satz 2]{Vol73} and \eqref{eq:comparison}, we have
  \begin{align}\label{eq:volkmann-comparison}
    q_{2}(t,\Re(u_1),0)\leq_{\cHplus}\Re(\psi_{2}(t,u_1,0)),\quad \text{for all } t<T\wedge T_{\psi_{2}}. 
  \end{align}
  We now prove that this implies that $T\leq T_{\psi_{2}}$. First, note   
  that whenever $u\in\cO$ and $v\in\cH$ are such that $u\leq_{\cHplus} v$,
  then $v\in\cO$, hence $\Re(\psi_{2}(t,u_1,0))$ does not approach the boundary of
  $\cO$ before $q_{2}(t,u_1,0)$ does. Next we show that also
  $\norm{\Re\psi_{2}(\cdot,u)}_2$ does not explode before
  $\norm{q_{2}(\cdot,u)}_2$. We show this as follows. From \eqref{eq:big-jumps}, \eqref{eq:small-jumps}, and \eqref{eq:volkmann-comparison}, we infer  
  \begin{align}\label{eq:existence-complex-solution-2}
  &\norm{\Re(\cR_{u_{1}}(\psi_{2}(t,u)))}_2\nonumber\\
  &\quad \leq
                                          \norm{{B}}_{\cL(\cH)}\norm{\Re(\psi_{2}(t,u))}_2+\frac{1}{2} M^2 \E^{2\omega t}\norm{D}_{\cL(\cH)}\norm{u_1}_H \nonumber\\
                                          &\qquad +M \E^{\omega t}\norm{D}_{\cL(\cH)}\norm{u_1}_H\norm{\Upsilon}_H \nonumber\\ 
                                        &\qquad + \norm{\int_{\cHplus\cap
                                          \set{0<\norm{\xi}_2\leq
                                          1}}\norm{\xi}_2^{-1}\mu(\D\xi)}_2\E^{\norm{q_{2}(t,\Re(u))}_2}\norm{\Re{\psi_{2}(t,u)}}_2\nonumber\\
    &\qquad
      +\norm{\mu(\cHplus\cap\set{\norm{\xi}_2>1})}_2+\norm{\int_{\cHplus\cap\set{\norm{\xi}_2>1})}\E^{-\langle\xi,q_{2}(t,\Re(u))\rangle}\frac{\mu(\D\xi)}{\norm{\xi}_2^{2}}}_2, 
  \end{align}
  where the constants $M \geq 1$ and $\omega \in \R$ are such that $\|S^*(t)\|_{\cL(H)} \leq M \E^{\omega t}$.
  Hence by an application of Gronwall's inequality we conclude that $\norm{\Re
    \psi_{2}(t,u)}_2$ is bounded as long as $\norm{q_{2}(t,u)}_2$ is bounded. Now, by
  Lemma~\ref{lem:growth-inequality}, we have 
  \begin{align*}
    \Re(\langle \overline{u},R(h,u)\rangle)\leq g(\Re(u))(1+\norm{u}_{\cH^\mathbb{C}}^{2})(1+\norm{h}_{H^\mathbb{C}}^{2}),
  \end{align*}
  where $g$ is a locally Lipschitz continuous function on $\cH$ and thus
  \begin{align*}
    \frac{\partial}{\partial t}\norm{\psi_{2}(t,u)}^{2}&=2\Re(\langle
                                                         \overline{\psi_{2}(t,u)},R(\psi_{1}(t,u),\psi_{2}(t,u))\rangle)\\
                                                       &\leq  g(\Re(\psi_{2}(t,u)))(1+\norm{\psi_{2}(t,u)}_{\cH^\mathbb{C}}^{2})(1+\norm{\psi_{1}(t,u)}_{H^\mathbb{C}}^{2})\,.
  \end{align*}
  Again by an application of Gronwall's inequality and since
  we already proved that $\Re(\psi_{2}(\cdot,u_1,0))$ does not explode before
  $q_{2}(\cdot,\Re(u_1),0)$ we conclude that $T_{\psi_{2}}\geq T$, which proves the assertion.
\end{proof}

\subsection{Analysis of the Riccati equations in the infinite-dimensional Wishart stochastic volatility model}
In the following proposition we show the existence of a unique solution to the
extended Riccati equations~\eqref{eq:riccati-finite-Wishart} 
up to its maximal lifetime.
\begin{proposition}\label{prop:extended-Riccati-existence-Wishart}
Let $n \in \mathbb{N}$, $Q$, $D \in \cS_1^+(H)$, $\Gamma \in H$ and $\fA \in \cL(H)$. Let $(\cA, \dom(\cA))$ be the generator of a strongly continuous semigroup $(S_t)_{t\geq 0}$ and let $(\cA^*, \dom(\cA^*))$ be its adjoint. Let $D \in \cL(H)$ be a positive
self-adjoint operator, and let $\Upsilon \in H$. Then for every $u=(u_{1},u_{2})\in H\times \cL_n$ there
exist a positive real number $T_{\sf q_{2}}$ and a unique solution $(P_n(\cdot,u), q_{1}(\cdot,u),
q_{2,n}(\cdot,u))$ to \eqref{eq:riccati-finite-phi}-\eqref{eq:riccati-finite-psi-2}
on $[0,T_{q_{2}})$.
\end{proposition}
\begin{proof}
Standard semigroup theory ensures that for any $u_{1}\in H$ and $T>0$ the unique solution to
\eqref{eq:riccati-finite-psi-1} is given by $q_{1,n}(t,u)=S_t^*u_{1}$ for $t\in
[0,T]$. We define $\cR_{u_1}^n (t, \cdot)\colon \cL_n \to \cL_n$ by 
\begin{align}\label{eq:Ru1}
\cR_{u_1}^n(t,h) &=  \tilde{\cR}_n(h)-\Pi_n\left(\tfrac{1}{2}\big(D^{1/2}S_t^*u_{1}\big)^{\otimes 2}+ D^{1/2}S_t^*u_{1}\otimes D^{1/2}\Upsilon\right)\Pi_n^*\,,
\end{align}
where 
$$\tilde{\cR}_n(h)=\Pi_n\left( \iota_n h\iota_n^*\fA^*+\fA \iota_n h \iota_n^* - \frac{1}{2}(\iota_n h \iota_n^*+ \iota_n h^T \iota_n^*) Q(\iota_n h \iota_n^*+ \iota_n h^T \iota_n^*)\right)\Pi_n^*\,.$$
Plugging $q_{1,n}(t,u)$ into \eqref{eq:riccati-finite-psi-2}, we thus obtain the equation
\begin{align*}
  \begin{cases}
    \frac{\partial q_{2,n}}{\partial t}(t,u)= \cR_{u_1}^n(t, q_{2,n}(t,u))\,,\\
     q_{2,n}(0,u)=u_{2}\,. 
  \end{cases}
\end{align*}
First observe that for $u, v \in \cL_n$
\begin{align*}
    &\norm{\tilde{\cR}_n(u) - \tilde{\cR}_n( v)}_1\\
    &\quad\leq  
    2\|u_n-v_n\|_{1} \|\cA\|_{1} + \frac{1}{2} \norm{(u_n+u_n^T)Q[(u_n+u_n^T)-(v_n+v_n^T)}_{1}\\
    &\qquad +\norm{[(u_n+u_n^T)-(v_n+v_n^T)]Q(v_n+v_n^T)}_{1} \\
    &\quad \leq C (\norm{u}_1 \vee \norm{v}_1)(\|u-v\|_1)\,,
\end{align*}
where $u_n = \iota_n u\iota^*_n$, $v_n = \iota_n v\iota^*_n$ and $C$ is a positive constant. It follows that 
$\cR_{u_{1}}^n(t,\cdot)$ is locally Lipschitz continuous on $\cL_n$ for every $t\geq 0$ and $u_{1}\in H$. Since $\cL_n \simeq \R^{n\times n}$ is finite-dimensional, standard ODE results (see
e.g.~\cite[Chapter 6, Proposition 1.1]{Mar76}) imply that for every $u_{2}\in \cL_n$ there exist a $T_{\sf q_{2}}>0$ and a unique solution
$q_{2,n}(\cdot,u)$ of~\eqref{eq:riccati-finite-psi-2} on $[0,T_{q_{2}})$ with $T_{q_{2}}=\infty$ or 
\begin{align*}
T_{q_{2}}=\liminf_{p\to\infty}\set{t\geq 0\colon \norm{ q_{2,n}(t,u)}_{1}\geq
  p}.
\end{align*}
By inserting $q_{2,n}(\cdot,u)$ in  \eqref{eq:riccati-finite-phi} and observing that
$F_n$ is continuous on $\cL_n$, the statement follows.
\end{proof}

\appendix

\section{Laguerre polynomials and evaluation of the basis function \texorpdfstring{$f_k$}{fk}}\label{sec:appendix1}
When computing the option prices, one has to calculate the integrals w.r.t. the basis function $f_k, k=1, \ldots, n$, see for instance, equations \eqref{eq:option_b0} and \eqref{eq:option1_b0}. This requires the evaluations of the basis function $f_k$, $k\geq 2$ which themselves are integrals of Laguerre polynomials with weight function $\E^{-\frac{1}{2}(\alpha+1)}$, see equation \eqref{eq:f_basis}. We start by introducing the Laguerre polynomials (see, for example, \cite[Chapter V]{freud2014orthogonal} for a good introduction to Laguerre polynomials).

The \textit{Laguerre polynomials} $L_n(x)$ are defined by
$$
L_n(x) = \frac{\E^x}{n!} \frac{\D^n}{\D x^n} \left( x^n \E^{-x} \right)\,, \qquad n\in \mathbb{N}.
$$
These polynomials are orthogonal with respect to the weight function $\E^{-x}$ on $[0, \infty)$, i.e.,
$$
\int_0^\infty L_n(x) L_m(x) \E^{-x} \, \D x = 
\begin{cases}
0 & \text{if } n \neq m, \\
1 & \text{if orthonormalized}, \\
\frac{1}{n!} \int_0^\infty \left[L_n(x)\right]^2 \E^{-x} \, \D x = \frac{1}{n!} & \text{(optional scaling)}.
\end{cases}
$$
They fulfill the so called \textit{three-term recurrence relation}
\begin{equation}\label{eq:reccurence}
L_{n+1}(x) = \frac{(2n + 1 - x) L_n(x) - n L_{n-1}(x)}{n + 1}\,,\qquad n\geq 2
\end{equation}

\noindent with initial conditions
$$
L_0(x) = 1, \quad L_1(x) = 1 - x\,.
$$
The derivative of the standard Laguerre polynomial $L_n(x)$ satisfies the identity
\begin{equation}\label{eq:derivative-laguerre}
x \frac{\D}{\D x} L_n(x) = n L_n(x) - n L_{n-1}(x)\,.
\end{equation}
We make use of these formulas for the Laguerre polynomials to present a fast evaluation approach for $f_k$, $k \geq 1$ by iterations. 
From \eqref{eq:reccurence}, it follows that, for $k\geq 1$
\begin{equation*}
\begin{aligned}
    f_{k+3}(x) &= \int_0^x L_{k+1}(s) \E^{-\frac{1}{2}(\alpha+1) s} \D s \\
    & = \frac{2k+1}{k+1}\int_0^x L_{k}(s)\E^{-\frac{1}{2}(\alpha+1) s}\D s - \frac{k}{k+1}\int_0^x L_{k-1}(s)\E^{-\frac{1}{2}(\alpha+1) s}\D s \\
    &\qquad - \frac{1}{k+1}\int_0^x sL_{k}(s)\E^{-\frac{1}{2}(\alpha+1) s}\D s \\
    & = \frac{2k+1}{k+1}f_{k+2}(x) - \frac{k}{k+1}f_{k+1} - \frac{1}{k+1}\int_0^x sL_{k}(s)\E^{-\frac{1}{2}(\alpha+1) s}\D s.
\end{aligned}   
\end{equation*}
Using \eqref{eq:derivative-laguerre}, we infer
\begin{equation*}
\begin{aligned}
    &\int_0^x sL_{k}(s)\E^{-\frac{1}{2}(\alpha+1) s}\D s \\
    &\quad = -\frac{2}{\alpha+1}sL_{k}(s)\E^{-\frac{1}{2}(\alpha+1) s}|_0^x + \frac{2}{\alpha+1} \int_0^x [L_k(s) + sL_{k}^\prime]\E^{-\frac{1}{2}(\alpha+1) s} \D s\\
    &\quad = -\frac{2}{\alpha+1}xL_{k}(x)\E^{-\frac{1}{2}(\alpha+1)x} + \frac{2}{\alpha+1}f_{k+2} + \frac{2}{\alpha+1} \int_0^x sL_{k}^\prime \E^{-\frac{1}{2}(\alpha+1) s} \D s\\
    &\quad  = -\frac{2}{\alpha+1}xL_{k}(x)\E^{-\frac{1}{2}(\alpha+1)x} + \frac{2}{\alpha+1}f_{k+2}\nonumber\\ 
    &\qquad + \frac{2}{\alpha+1} \int_0^x [ kL_k(s) - kL_{k-1}(s)] \E^{-\frac{1}{2}(\alpha+1) s} \D s \\
    &\quad  = -\frac{2}{\alpha+1}xL_{k}(x)\E^{-\frac{1}{2}(\alpha+1)x} + \frac{2}{\alpha+1}(k+1)f_{k+2} - \frac{2}{\alpha+1}kf_{k+1}\,.
\end{aligned}
\end{equation*}
Finally, we have the iteration for the basis functions $f_k, k\geq 1$
\begin{equation*}
\begin{aligned}
    f_{k+3}(x) &= \frac{2k+1}{k+1}f_{k+2}(x) - \frac{k}{k+1}f_{k+1} \\
    &+ \frac{1}{k+1}\frac{2}{\alpha+1}\left[xL_{k}(x)\E^{-\frac{1}{2}(\alpha+1)x} - (k+1)f_{k+2} + kf_{k+1}\right]\,.
\end{aligned}
\end{equation*}


\section{Numerical method for solving the Riccati equations in the option pricing of Wishart model}\label{sec:riccat_ws}
In the characteristic function and the option pricing equation, one has to compute the quantities
$\langle X_0, q_1(T, u)\rangle_w$, $\tr(Y_0q_2(T, u))$ and $P(T, u)$.
The first term is explicit because $q_1(t, u)$ solves the same linear equation as $\psi_1(t,u)$ in \eqref{eq:riccati_pure_jump1}. The scalar Riccati term satisfies
\begin{equation*}
    \frac{\partial P}{\partial t}(t, u) = n\tr(Qq_2(t, u)), \qquad P(0,u)=0.
\end{equation*}

For the numerical Wishart experiment of Section~\ref{sec:numerics}, we work with the diagonal parameterization
\begin{align*}
Q &= \sum_{k\ge1} q_k\, f_k\otimes f_k,\qquad
\mathbb{A} = -\sum_{k\ge1} a_k\, f_k\otimes f_k,\qquad
D = \frac{1}{2}\sum_{k\ge1}\frac{1}{k^2}\,f_k\otimes f_k,\\
\Upsilon^\vartheta &= -\tfrac12 S^*(\vartheta)h_0,\qquad
D^{1/2}=\frac{1}{\sqrt2}\sum_{k\ge1}\frac{1}{k}\,f_k\otimes f_k,
\end{align*}
with $q_k>0$ and $a_k>0$. For $u_1=(\nu+\mathrm{i}\lambda)u_\vartheta$, define
\begin{align*}
b_k(t,u) &:= \langle D^{1/2}q_1(t,u),f_k\rangle_w
        = \frac{\nu+\mathrm{i}\lambda}{\sqrt2\,k}\,f_k(t+\vartheta),\\
c_k^\vartheta &:= \langle D^{1/2}\Upsilon^\vartheta,f_k\rangle_w
        = -\frac{1}{2\sqrt2\,k}\,\langle f_k,S^*(\vartheta)h_0\rangle_w.
\end{align*}
Moreover, let
\begin{equation*}
F_{ij}(t,u):=\langle q_2(t,u)f_i,f_j\rangle_w,\qquad i,j\ge1.
\end{equation*}
Note that the coefficients $F_{ij}$ need not satisfy $F_{ij}=F_{ji}$, because the mixed term in~\eqref{eq:R-Wishart} is not symmetry preserving.
Since $q_2$ solves the Wishart Riccati equation associated with \eqref{eq:R-Wishart},
\begin{align*}
\partial_t q_2(t,u)
&= q_2(t,u)\mathbb{A}^*+\mathbb{A}q_2(t,u)
   -\tfrac12(D^{1/2}q_1(t,u))^{\otimes2}
   -D^{1/2}q_1(t,u)\otimes D^{1/2}\Upsilon^\vartheta \\
&\quad -\tfrac12(q_2(t,u)+q_2(t,u)^T)Q(q_2(t,u)+q_2(t,u)^T),
\end{align*}
we obtain, for all $i,j\ge1$,
\begin{equation}\label{eq:q_2_prod_infty}
\begin{aligned}
\frac{\partial}{\partial t}F_{ij}(t,u)
&= -(a_i+a_j)\,F_{ij}(t,u)
   -\frac12\, b_i(t,u)b_j(t,u)
   - b_i(t,u)c_j^\vartheta \\
&\quad -\frac12\sum_{k=1}^\infty q_k\,
   \big(F_{ik}(t,u)+F_{ki}(t,u)\big)\big(F_{kj}(t,u)+F_{jk}(t,u)\big).
\end{aligned}
\end{equation}

To implement the scheme, we truncate the system to $H_n=\mathrm{span}\{f_1,\dots,f_n\}$. Define
\begin{equation*}
\mathbf F_n(t,u):=\big(F_{ij}(t,u)\big)_{1\le i,j\le n},\qquad
\mathbf A_n:=\mathrm{diag}(a_1,\dots,a_n),\qquad
\mathbf Q_n:=\mathrm{diag}(q_1,\dots,q_n),
\end{equation*}
\begin{equation*}
\mathbf b_n(t,u):=\big(b_1(t,u),\dots,b_n(t,u)\big)^\top,\qquad
\mathbf c_n^\vartheta:=\big(c_1^\vartheta,\dots,c_n^\vartheta\big)^\top.
\end{equation*}
An explicit Euler step with mesh size $\delta>0$ yields the recursion
\begin{equation}\label{eq:q_2_prod_matrix}
\begin{aligned}
\mathbf F_n(t+\delta,u)
&\approx \mathbf F_n(t,u)-\delta\big(\mathbf A_n\mathbf F_n(t,u)+\mathbf F_n(t,u)\mathbf A_n\big)\\
&\quad -\frac{\delta}{2}\,\mathbf b_n(t,u)\mathbf b_n(t,u)^\top
      -\delta\,\mathbf b_n(t,u)(\mathbf c_n^\vartheta)^\top \\
&\quad -\frac{\delta}{2}\big(\mathbf F_n(t,u)+\mathbf F_n(t,u)^\top\big)
      \mathbf Q_n
      \big(\mathbf F_n(t,u)+\mathbf F_n(t,u)^\top\big).
\end{aligned}
\end{equation}
The initial condition is
\begin{equation*}
\mathbf F_n(0,u)=\big(\langle u_2f_i,f_j\rangle_w\big)_{1\le i,j\le n}.
\end{equation*}
In the first-moment experiment of Section~\ref{sec:numerics}, we take $u_2=0$, hence $\mathbf F_n(0,u)=0$.

The trace term is then approximated by
\begin{equation*}
\tr(Qq_2(t,u))\approx \sum_{k=1}^n q_k\,F_{kk}(t,u)=\tr(\mathbf Q_n\mathbf F_n(t,u)),
\end{equation*}
and therefore
\begin{equation*}
P_n(t+\delta,u)\approx P_n(t,u)+\delta\,n\,\tr(\mathbf Q_n\mathbf F_n(t,u)),\qquad P_n(0,u)=0.
\end{equation*}
Together with the explicit formula for $q_1$, this gives the finite-rank transform approximation used in the Wishart numerical experiment. In Section~\ref{sec:numerics} we further specialize to $q_k=a_k=1/k^2$.

\bibliographystyle{acm}
\bibliography{literatur}

\end{document}